\documentclass[reqno,11pt]{amsart}

\IfFileExists{mymtpro2.sty}{%
  \usepackage[subscriptcorrection]{mymtpro2}
}{}

\usepackage{a4,amssymb}

\usepackage{amsmath}

\newtheorem{theorem}{Theorem}[section]
\newtheorem{lemma}{Lemma}[section]
\newtheorem{corollary}{Corollary}[section]
\newtheorem{prop}{Proposition}[section]
\newtheorem{remark}[theorem]{Remark}

\newtheorem{definition}[theorem]{Definition}

\newcommand{\labelnummer}{\mbox{\normalfont (\roman{numcount})}}%

\makeatletter

  {\let\curlabelspeicher\@currentlabel%
    \begin{list}{\labelnummer}%
      {\usecounter{numcount}\leftmargin0pt%
        \topsep0.5ex\partopsep2ex\parsep0pt\itemsep0ex\@plus1\p@%
        \labelwidth2.5em\itemindent3.5em\labelsep1em%
      }%
    \let\saveitem\item%
    \def\item{\saveitem%
      \def\@currentlabel{{\upshape\curlabelspeicher}$\,$\labelnummer}}%
    \let\savelabel\label%
    \def\label##1{\savelabel{##1}%
      \@bsphack%
        \ifmmode\else%
          \protected@write\@auxout{}%
          {\string\newlabel{##1item}{{\labelnummer}{\thepage}}}%
        \fi%
      \@esphack%
    }%
  }{\end{list}}%

\renewcommand{\appendix}{\def\thesection{\textsc{Appendix}}}

 \let\leq\le
 \let\geq\ge

 \let\Im\undefined

\DeclareMathOperator{\Im}{Im}

\DeclareMathOperator{\tr}{tr\kern1pt}

\makeatletter

\newif\ifper\pertrue
\def\per{.}

\def\bti{\@ifnextchar[\bbti\bbbti}
\def\bbti[#1]#2{#2, #1.}
\def\bbbti#1{#1.}

\def\z{\@ifnextchar[\zz\zzz}
\def\zz[#1]#2#3#4#5{\perfalse\emph{#2} \textbf{#3}, #4 (#5) [#1]}
\def\zzz#1#2#3#4{\emph{#1} \textbf{#2}, #3 (#4)\ifper\per\fi\pertrue}

\def\pub{\@ifstar\pubstar\pubnostar}
\def\pubnostar{\@ifnextchar[\@@pubnostar\@pubnostar}
\def\@@pubnostar[#1]#2#3#4{#2, #3, #4, #1\ifper\per\fi\pertrue}
\def\@pubnostar#1#2#3{#1, #2, #3\ifper\per\fi\pertrue}
\def\pubstar[#1]#2#3#4{\perfalse #2, #3, #4 [#1]\pertrue}

\makeatother

\newcommand{\bel}{\begin{equation} \label}
\newcommand{\ee}{\end{equation}}

\def\beq{\begin{equation}}
\def\eeq{\end{equation}}
\newcommand{\bea}{\begin{eqnarray}}
\newcommand{\eea}{\end{eqnarray}}
\newcommand{\beas}{\begin{eqnarray*}}
\newcommand{\eeas}{\end{eqnarray*}}

{

\newcommand{\R}{\mathbb{R}}

\newcommand{\Z}{\mathbb{Z}}

\newcommand{\N}{\mathbb{N}}

\newcommand{\C}{\mathbb{C}}
\newcommand{\E}{\mathbb{E}}

\begin{document}

\title[Dependence of the density of states on the probability distribution - part II]{Dependence of the density of states on the probability distribution - part II: Schr\"odinger operators on $\R^d$ and non-compactly supported probability measures}

\author[P.\ D.\ Hislop]{Peter D.\ Hislop}
\address{Department of Mathematics,
    University of Kentucky,
    Lexington, Kentucky  40506-0027, USA}
\email{peter.hislop@uky.edu}

\author[C.\ A.\ Marx]{Christoph A.\ Marx}
\address{Department of Mathematics,
Oberlin College,
Oberlin, Ohio 44074, USA}
\email{cmarx@oberlin.edu}


\begin{abstract}
We extend our results in \cite{hislop_marx_1} on the quantitative continuity properties, with respect to the single-site probability measure, of the density of states measure and the integrated density of states
for random Schr\"odinger operators. For lattice models on $\Z^d$, with $d \geq 1$, we treat the case of non-compactly supported probability measures with finite first moments. For random Schr\"odinger operators on $\R^d$, with $d \geq 1$, we prove results analogous to those in \cite{hislop_marx_1} for compactly supported probability measures. The method of proof makes use of the Combes-Thomas estimate and the Helffer-Sj\"ostrand formula.
\end{abstract}

\maketitle \thispagestyle{empty}

\tableofcontents

\vspace{.2in}



\section{Introduction} \label{sec_intro}
\setcounter{equation}{0}

This is the second of a pair of papers, the first being \cite{hislop_marx_1}, in which we study the dependency of almost sure quantities such as the density of states measure and the integrated density of states of random Schr\"odinger operators on the single-site probability measure.

%

For lattice models on $\Z^d$, we consider the formal Hamiltonian on $\ell^2 (\Z^d)$ with a random potential constructed from finite-rank projections and independent, identically distributed ($iid$) random variables,
\beq  \label{eq_discreteschrodop}
H_\omega =  - \Delta + \sum_{j \in \mathcal{J}} \omega_j P_j ~\mbox{.}
\eeq
Here, $\Delta$ is the usual finite-difference Laplacian on $\Z^d$ and the elements of $\omega = (\omega_j)_{j \in \mathcal{J}}$ are distributed according to a common Baire probability measure $\nu \in \mathcal{P}(\mathbb{R})$ (``the single-site measure'') with projections $\{P_j, ~j \in \mathcal{J}\}$ forming a complete family of orthogonal projections with common rank $N \in \mathbb{N}$ indexed by a countable set $\mathcal{J}$; a precise definition of the model (\ref{eq_discreteschrodop}) is given in hypothesis [\textbf{Disc}] of section \ref{sec_mainresults}. For the usual Anderson model, the projections are rank-one, i.e. $N=1$. Models for $N > 1$ arise for instance in the study of multi-dimensional random polymers \cite{deBievre_Germinet_2000, SchulzBaldes_Jitomirskaya_Stolz_CMP_2003, DamanikSimsStolz_JFunAnal_2004}. More generally, as in part one \cite{hislop_marx_1} of this two-part series (see section 6 in \cite{hislop_marx_1}), our treatment allows us to consider more general finite-difference operators in more general settings, such as graphs. In particular, our methods apply to the Anderson model on the Bethe lattice.


The main object of interest is the {\em{density of states measure}} (DOSm) associated with (\ref{eq_discreteschrodop}), defined as a spectral average
\begin{equation} \label{eq_DOSm_discr}
n_\nu^{(\infty)}(f):= \dfrac{1}{N} \mathbb{E}_{\nu^{(\infty)}} \mathrm{Tr} \left( P_0 f(H_\omega) P_0  \right) ~\mbox{, } f \in \mathcal{C}_c(\mathbb{R}) ~\mbox{}
\end{equation}
with respect to the infinite product measure $\nu^{(\infty)} := \bigotimes_{j \in \mathcal{J}} \nu$. We also consider the {\em{integrated density of states}} (IDS), $N_\nu(E) :=n_\nu^{(\infty)}((-\infty, E])$, the right-continuous cumulative distribution of $n_\nu^{(\infty)}$.

In our recent paper \cite{hislop_marx_1}, we showed that for single-site probability measures supported on a fixed {\em{compact}} set $[-C,C]$, the map
\begin{equation}
\nu \mapsto n_\nu^{(\infty)} ~\mbox{ is H\"older continuous }
\end{equation}
in the weak-$^*$ topology associated with the dual of $\mathcal{C}([-C,C])$.

The purpose of this article is twofold. First, we extend the result for lattice models to single-site measures $\nu$ which are not necessarily compactly supported. Second, and more importantly, we prove a similar result as in \cite{hislop_marx_1} for the continuum analogue of (\ref{eq_discreteschrodop}), i.e.
\begin{equation} \label{eq_contschrodop}
H_\omega = - \Delta + \sum_{j \in \mathbb{Z}^d} \omega_j \phi( \cdot - j) ~\mbox{ on $L^2(\mathbb{R}^d)$ .}
\end{equation}
Here, $-\Delta \geq 0$ is the nonnegative Laplacian on $\mathbb{R}^d$, $0 \leq \phi \leq 1$ is a fixed continuous function supported in a compact neighborhood of the origin, and $\omega_j$ are iid potentials with underlying single-site probability measure $\nu \in \mathcal{P}(\mathbb{R})$ as above. Unlike the discrete case \eqref{eq_discreteschrodop}, the perturbation is given by the operator of  multiplication by $\phi(\cdot - j)$ that has infinite rank.  In the continuum case, the definition of the DOSm in (\ref{eq_DOSm_discr}) is replaced by
\begin{equation} \label{eq_DOSm_cont}
n_\nu^{(\infty)}(f):= \mathbb{E}_{\nu^{(\infty)}} \mathrm{Tr}\left(  \chi_0 f(H_\omega)  \chi_0  \right) ~\mbox{, } f \in \mathcal{C}_c(\mathbb{R}) ~\mbox{,}
\end{equation}
where $\chi_0:= \chi_{\Lambda_0}$ is the characteristic function of the unit cube in $\mathbb{R}^d$ centered at the origin,
\begin{equation}
\Lambda_0:= \{ x \in \mathbb{R}^d ~:~ \Vert x \Vert_\infty \leq 1/2 \} ~\mbox{,}
\end{equation}
and $\mathcal{C}_c(\mathbb{R})$ denotes the set of continuous functions with compact support. In general, for a function space $X$, we denote the subset of functions of compact support by $X_c$.

We refer the reader to the introduction in the first paper \cite{hislop_marx_1} for an introduction to the problem, including a historical survey of related results, and the motivation for these works. Many of the applications in that paper can be extended to the more general situations treated here.

\subsection{Context and outline of the proof} \label{sec_intro_context}

Given a Borel set $A \subseteq \mathbb{R}$, let $\mathcal{P}(A)$ denote the Baire probability measures supported on $A$. Our earlier result in \cite{hislop_marx_1} established that for the lattice model defined in (\ref{eq_discreteschrodop}) and with underlying single-site measures in the space $\mathcal{P}([-C,C])$, weak-$^*$ convergence of a sequence $(\nu_\alpha)$ to a measure $\nu$ in $\mathcal{P}([-C,C])$, implies that for all $f \in \mathrm{Lip}([-C-2d, C+2d])$, one has
\begin{equation} \label{eq_previousresult}
\vert n_{\nu_\alpha}^{(\infty)}(f) - n_\nu^{(\infty)}(f) \vert \leq \gamma \Vert f \Vert_{\mathrm{Lip}} ~d_w(\nu_\alpha, \nu)^\kappa ~\mbox{, } \forall \alpha \geq \alpha_0 ~\mbox{,}
\end{equation}
for absolute constants $\alpha_0 \in \mathbb{N}$, $\gamma > 0$, and $0 < \kappa < 1$, only depending on the dimension $d$ and the rank $N$. Here, $d_w( . , .)$ in (\ref{eq_previousresult}) is an appropriate metric which metrizes the weak-$^*$ topology on $\mathcal{P}([-C,C])$ and $\mathrm{Lip}([-C-2d,C+2d])$ is the Banach space of Lipschitz functions equipped with the usual Lipschitz norm $\Vert . \Vert_{\mathrm{Lip}}$, see also section \ref{sec_metrizVague} for precise definitions.

The proof of the quantitative continuity result (\ref{eq_previousresult}) in \cite{hislop_marx_1} relied on a strategy involving the following two steps, the first of which crucially depended  on the compactness of the support of the single-site measures, which causes the operator in (\ref{eq_discreteschrodop}) to be effectively bounded. The two steps are:
\begin{description}
\item[Step 1 - Finite range reduction] The finite range reduction is based on polynomial approximation of the functions $f$ in (\ref{eq_previousresult}), which allows us to first only consider the effects of varying the probability measure on a {\em{finite}} subdomain of $\mathbb{Z}^d$. The key idea underlying this finite-range reduction is that for every {\em{polynomial}} $f$, the map
\begin{equation}\label{eq:tr-map1}
(\omega_n) := \omega \mapsto \mathrm{Tr}(P_0 f(H_\omega) P_0)
\end{equation}
depends on only {\em{finitely}} many components $\omega_n$ of $\omega$.

\vspace{0.1 cm}
\item[Step 2 - Lipschitz property] This is the Lipschitz continuity of the map in \eqref{eq:tr-map1} with respect to a single random variable $\omega_n$ while keeping all the other variables $\omega_j$ with $j \neq n$ fixed.
\end{description}

The modified proof strategy presented in this article will allow us to extend above-mentioned continuity result for the lattice model (\ref{eq_discreteschrodop}) to include non-compactly supported single-site measures, as well as to treat the continuum models in (\ref{eq_contschrodop}). 
We mention that for the continuum models (\ref{eq_contschrodop}), 
we will still need to restrict our results to compactly supported measures.

The approach we take here is based on a non-trivial modification of step 1 (``finite range reduction''), replacing polynomials with resolvents and truncating to a finite number of random variables by taking advantage of the Combes-Thomas estimate. Working with resolvents will, however, come at a price: the singularity of resolvents on the real axis will have to be compensated for by assuming higher regularity of the functions $f$ in (\ref{eq_previousresult}), thereby limiting the continuity result we obtain to functions $f \in \mathcal{C}_c^M(\mathbb{R})$ for some $M = M(d) > 1$, instead of merely Lipschitz as in (\ref{eq_previousresult}).

Specifically, for the lattice model in (\ref{eq_discreteschrodop}), while the map
\begin{equation}\label{eq:resolvent-map1}
(\omega_n)=\omega \mapsto \mathrm{Tr}(P_0 (H_\omega - z)^{-1} P_0) ~\mbox{, } z \in \mathbb{C} \setminus \mathbb{R} ~\mbox{,}
\end{equation}
no longer depends on only finitely many components $\omega_n$ of $\omega$, the Combes-Thomas estimate will imply that for every single-site measure $\nu \in \mathcal{P}(\mathbb{R})$ with finite first moment
\begin{equation}
\mu_1[\nu]:= \int \vert x \vert ~\mathrm{d} \nu(x) ~\mbox{,}
\end{equation}
the {\em{averaged}} contribution to \eqref{eq:resolvent-map1} of the random variables $\omega_n$ with $\Vert n \Vert_\infty > L$ decays as
\begin{equation} \label{eq_finiterangeestim}
 \lesssim \dfrac{ \mu_1[\nu] }{\vert \mathrm{Im} z \vert^{2 + d} } \mathrm{e}^{-c \vert \mathrm{Im} z \vert L} ~\mbox{.}
\end{equation}
The Helffer-Sj\"ostrand functional calculus then allows us to control the singularity of (\ref{eq_finiterangeestim}) near the real axis by working with functions $f$ of regularity of at least $2 + d = M(d)$. Since the key ingredient in modified step 1 is the Combes-Thomas estimate, the argument is equally applicable to the continuum models in (\ref{eq_contschrodop}).

Finally, for the continuum model (\ref{eq_contschrodop}), the proof of both step 1 and step 2 has to be further revised since the projector $\chi_0$ in (\ref{eq_DOSm_cont}) is not finite rank. We will use appropriate regularizations of the operator (\ref{eq_contschrodop}) relying on the relative compactness of $\chi_0$ with respect to appropriate powers (depending on $d$) of the resolvent. The latter will result in an increase of the required regularity $M$. Moreover, our proof of the Lipschitz property (step 2) for the continuum case will require us to assume that the single site measures are compactly supported. Although the finite range reduction (step 1) for the continuum models does not require compactly supported single site measures, but only finiteness of a certain number of moments, we will assume that the single site probability measures for continuum models are compactly supported for simplicity.




\subsection{Vague and weak topology} \label{sec_metrizVague}

To examine the continuity of the map
\begin{equation} \label{eq_mainmap}
\nu \mapsto n_{\nu}^{(\infty)} ~\mbox{,}
\end{equation}
while possibly allowing for a non-compactly supported single site measure $\nu$, we need to agree on suitable topologies for both the domain and the codomain of the map in (\ref{eq_mainmap}). To this end, given Borel set $A \subseteq \mathbb{R}$, we first recall the following two topologies on $\mathcal{P}(A)$:
\begin{definition}
Let $(\mu_\alpha)_{\alpha \in \mathbb{N}}$ be a sequence of measures in $\mathcal{P}(A)$.
\begin{itemize}
\item[(i)] The sequence $(\mu_\alpha)_{\alpha \in \mathbb{N}}$ converges in {\em{vague topology}} if there exists a measure $\mu \in \mathcal{P}(A)$ such that for all continuous, compactly supported functions $f \in C_c(\mathbb{R})$, one has $\mu_\alpha(f) \to \mu(f)$. We denote the topological space $\mathcal{P}(A)$ equipped with the vague topology as the pair $(\mathcal{P}(A), \mathcal{V})$.
\item[(ii)] The sequence $(\mu_\alpha)_{\alpha \in \mathbb{N}}$ converges in {\em{weak topology}} if there exists a measure $\mu \in \mathcal{P}(A)$ such that for all continuous and bounded functions $f \in C_b(\mathbb{R})$, one has $\mu_\alpha(f) \to \mu(f)$. We denote the topological space $\mathcal{P}(A)$ equipped with the weak topology as the pair $(\mathcal{P}(A), \mathcal{W})$.
\end{itemize}
\end{definition}
\begin{remark} \label{eq_topology}
It is clear that, in general, the weak topology is stronger than the vague topology. We note, however, that if $A \subseteq \mathbb{R}$ is compact, the two topologies agree. In particular, in view of our earlier results in \cite{hislop_marx_1}, if $A=[-C,C]$ for $0<C<+\infty$, both weak and vague topology coincide with the weak-$^*$ topology associated with the dual of $\mathcal{C}([-C,C])$.
\end{remark}

Taking into account that the important characterization of the DOSm as an almost sure limit is formulated the vague topology, see e.g. \cite{CyconFroeseKirschSimon_book_1987} Sec. 9.2 therein, it is natural to use the {\em{vague topology}} for the {\em{codomain}} space of the map in (\ref{eq_mainmap}).

On the other hand, in view of the ``Lipschitz property'' outlined as step 2 in Sec. \ref{sec_intro_context}, we observe that for the lattice model (\ref{eq_discreteschrodop}), a given function $f \in \mathcal{C}_c(\mathbb{R})$, and an arbitrary lattice site $n \in \mathbb{Z}^d$, while the map
\begin{equation}
\omega_n \mapsto \mathrm{Tr}(P_0 f(H_\omega) P_0)
\end{equation}
does define a continuous and bounded function with $\Vert \mathrm{Tr}(P_0 f(H_\omega) P_0) \Vert_\infty \leq N \Vert f \Vert_\infty$, it is in general not compactly supported. Hence, to allow for non-compactly supported single-site measures $\nu \in \mathcal{P}(\mathbb{R})$, we will equip the {\em{domain}} of the map in (\ref{eq_mainmap}) with the {\em{weak topology}}.

As in \cite{hislop_marx_1}, we will use the well-known fact (see e.g. \cite[Theorem 12]{Dudley_1966}) that for each fixed Borel set $A \subseteq \mathbb{R}$, the weak topology on $\mathcal{P}(A)$ is metrizable by a metric derived from the Lipschitz dual:
\begin{equation}\label{eq:LipMetric1}
d_w(\mu, \nu):= \sup \left\{ \left\vert \mu(f) - \nu(f) \right\vert  ~:~ f \in \mathrm{Lip}_b(A) ~\mbox{with } \Vert f \Vert_{\mathrm{Lip}} \leq 1  \right\} ~\mbox{,}
\end{equation}
where $\mathrm{Lip}_b(A)$ is the space of bounded Lipschitz functions on $A$, equipped with the usual norm:
\begin{equation} \label{eq_lipschitznorm}
\Vert f \Vert_{\mathrm{Lip}} := \Vert f \Vert_\infty + \sup_{x \neq y \in A} \dfrac{\vert f(x) - f(y) \vert}{\vert x - y \vert} =: \Vert f \Vert_\infty + L_f ~\mbox{.}
\end{equation}

Finally, to account for finite moment conditions on the single-site measure, given a Borel set $A \subseteq \mathbb{R}$, $C>0$ and $p \in \mathbb{N}$, we let
\begin{align} \label{eq_probspace_boundedfinitemoment}
\mathcal{P}_{p; C}(A):=\{ \nu \in \mathcal{P}(A) ~:~ \max_{1 \leq l \leq p} \mu_l[\nu] \leq C \} ~\mbox{, }
\end{align}
denote all probability measures $\nu$ on $A$ with moments $\mu_l[\nu]:= \int \vert x \vert^l ~\mathrm{d} \nu(x) \leq C$ of degree $l$ for all $1 \leq l \leq p$.


\subsection{Statements of the main results} \label{sec_mainresults}

In this paper, we limit ourselves to continuity properties of the DOSm and IDS with respect to the single-site probability measure. Using the methods presented here, other results analogous to those in \cite{hislop_marx_1} can be proven for the models treated here.

\vspace{.1in}
\noindent
{\bf Discrete models on $\ell^2 (\Z^d).$}
\vspace{.1in}

The precise description of the random Schr\"odinger operator with finite-rank potentials, formally given in (\ref{eq_discreteschrodop}), is as follows:
\begin{description}
\item [{[Disc]}]  Fix $K \in \mathbb{N}$ and set $N:=K^d$. The discrete Hamiltonian on $\ell^2 (\Z^d)$ has the form (\ref{eq_discreteschrodop}) where the components of $\omega := \{ \omega_j \}_{ j \in {\mathcal{J}}}$ are iid random variables, distributed according to a common probability measure $\nu \in \mathcal{P}(\mathbb{R})$. Here, the index set ${\mathcal{J}}$ is a lattice $K \Z^d$. The rank $N$ projection $P_0$ projects onto the $N = K^d$ sites in the cube $[0, K-1]^d \subset \Z^d$ at the origin. The orthogonal projections $\{ P_j ~|~ j \in \mathcal{J} \}$ are generated by translation of the single rank $N$ projection $P_0$, i.e., for each $j \in \mathcal{J}$, we define $P_j = U_j P_0 U_j^{-1}$, where $U_j$ denotes the unitary on $\ell^2 (\Z^d)$ given by $f(x-j) = (U_jf)(x)$.
\end{description}
 \vspace{.1in}
We then have:
\begin{theorem}\label{thm:main-lattice1}
Let $H_\omega$ be the discrete random Schr\"odinger operators described in hypothesis {\bf[{Disc}]}.  For $C > 0$, we recall from (\ref{eq_probspace_boundedfinitemoment}) that $\mathcal{P}_{1;C}(\mathbb{R})$ denotes the space of single-site probability measures supported on $\R$ with finite first moments bounded by $C$.  
\begin{itemize}
\item[(i.)] For each $C>0$, $A \in \mathbb{R}$, and $E \in \mathbb{R}$, both the maps
\begin{align} \label{eq_mainmap-discr1}
\mathcal{N}: (\mathcal{P}_{1;C}(\mathbb{R}), \mathcal{W}) \to (\mathcal{P}(\mathbb{R}), \mathcal{V}) ~, ~\nu \mapsto n_{\nu}^{(\infty)} ~\mbox{,} \nonumber \\
\mathcal{I}:  (\mathcal{P}_{1;C}([A, + \infty)), \mathcal{W}) \to \mathbb{R} ~, ~\nu \mapsto N_\nu(E) ~\mbox{,}
\end{align}
are continuous.

\item[(ii.)] The modulus of continuity of the map $\mathcal{N}$ in (\ref{eq_mainmap-discr1}) is quantified by the following: There exist $\rho_1>0$ and a degree of regularity $M_1 \in \mathbb{N}$, only depending on $d$, such that for each given $C>0$, there is a constant $C_{1}$, depending on $d$, $N$, and $C$, so that for all measures $\nu_1, \nu_2 \in \mathcal{P}_{1;C}(\mathbb{R})$ with $d_w(\nu_1, \nu_2) < \rho_1$ and all functions $f \in \mathcal{C}_c^{M_1}(\mathbb{R})$ with $\mathrm{supp}(f) \subseteq [-r,r]$ and $r \geq 1$, one has
\begin{equation}
| n_{\nu_1}^{(\infty)}(f) - n_{\nu_2}^{(\infty)} (f)| \leq C_{1} r^{M_{1}} \Vert f \Vert_{\mathcal{C}^{M_1}} ~\cdot d_w(\nu_1,\nu_2)^{\frac{1}{1+d}} ~\mbox{.}
\end{equation}
\item[(iii.)] For the map $\mathcal{I}$ in (\ref{eq_mainmap-discr1}), the modulus of continuity is quantified by the following: For $\rho_1>0$ as in part (ii.), for each $C>0$, $A \in \mathbb{R}$, all measures $\nu_1, \nu_2 \in \mathcal{P}_{1;C}([A,+\infty))$ with $d_w(\nu_1, \nu_2) < \rho_1$, and all $E_0 \in \mathbb{R}$, there exists a constant $C_{2; A, E_0}$, depending on $A$, $E_0$, $C$, $d$, and the rank $N$, so that if $E \leq E_0$, one has
\begin{equation}
| N_{\nu_1} (E) - N_{\nu_2}(E) | \leq \frac{C_{2; A, E_0}}{\log \left(   \frac{1}{ d_w(\nu_1,\nu_2)    }    \right)} ~\mbox{.}
\end{equation}
\end{itemize}
\end{theorem}
Here, for $\beta \in \mathbb{N}$ and $f \in \mathcal{C}_c^\beta$, we denote
\begin{equation} \label{eq_cbetanorm}
\Vert f \Vert_{\mathcal{C}^\beta}:= \sum_{k=0}^\beta \Vert f^{(k)} \Vert_\infty ~\mbox{.}
\end{equation}
\begin{remark} \label{rem_mainthm_disc}
\item[(i).] We note that our proof shows that one can take $M_1 = d+3$ and the radius $\rho_1 = \left( \frac{2}{3} \right)^{1 + d}$.
\item[(ii.)] In view of the statements about the IDS, considering measures $\nu$ whose support is lower semi-bounded, $\mathrm{supp}(\nu) \subseteq [A, +\infty)$ for some $A \in \mathbb{R}$, implies that the spectrum of $H_\omega$ is contained in $[-2d - \vert A \vert, +\infty)$, for all $\omega \in \Omega$. This ensures that, on the spectrum, the step-function $\chi_{(-\infty, E)}$ can be approximated uniformly by {\em{compactly supported}} functions, in agreement with the vague topology used for the codomain of the map $\mathcal{N}$.
\end{remark}

\vspace{.1in}
\noindent

{\bf Continuum models on $L^2 (\R^d)$.}
\vspace{.1in}

For the random Schr\"odinger operator in $\mathbb{R}^2$ given formally by (\ref{eq_contschrodop}), we add the following hypotheses:
\begin{description}
\item [{[Cont]}] Consider the random Schr\"odinger operator in (\ref{eq_contschrodop}), acting on $L^2(\mathbb{R}^d)$, where the random potential has the form
\beq\label{eq:random-pot-cont1}
V_\omega (x) = \sum_{j \in \Z^d} ~ \omega_j \phi (x - j) =: \sum_{j \in \Z^d} ~ \omega_j \phi_j (x),
\eeq
and $0  \leq \phi (x) \leq 1$ is a $C_c^{k_v}(\mathbb{R}^d)$-function with compact support in a neighborhood of the origin in $\mathbb{R}^d$. The required degree of regularity $k_v \geq 0$ will be depend on the dimension. Without loss of generality we will assume that $0  \leq \phi (x) \leq 1$ is supported in the unit cube $\Lambda_0$, centered at the origin, which in particular implies $\phi = \phi \cdot \chi_0$ and $\phi \leq \chi_0$.
\end{description}

We will assume that the support of the single-site probability measure is compact and contained in $[-C,C]$ for a finite constant $C>0$. In this case, the spectrum of $H_\omega$ is contained in the half-line $[-C, \infty)$.
For these continuum models, we prove the following result. As mentioned earlier at the end of section \ref{sec_intro_context}, our proof requires us to restrict the result to compactly supported single site measures. In view of the topologies used for the qualitative continuity statement in part (i) of Theorem \ref{thm:main-cont1}, we thus recall remark \ref{eq_topology} for the case of compactly supported measures.

\begin{theorem} \label{thm:main-cont1}
Let $H_\omega$ be the continuum model described in hypothesis {\bf[{Cont}]} with degree of regularity $k_v \geq 0$ such that $k_v > \max\{ d - 2 ; \frac{4 + 2d}{3} \}$. Given $0< C < + \infty$, let $\mathcal{P}([-C,C])$ denote the space of probability measures supported on $[-C,C]$.
\begin{itemize}
\item[(i.)] For each $C>0$, the map
\begin{align} \label{eq_mainmap-cont1}
\mathcal{N}: (\mathcal{P}([-C,C]) , \mathcal{W}) \to (\mathcal{P}(\R), \mathcal{V}) ~, ~\nu \mapsto n_{\nu}^{(\infty)} ~\mbox{}
\end{align}
is continuous.
Moreover, for each given $E \in \mathbb{R}$, the map
\begin{align} \label{eq_mainmap-IDScont1}
\mathcal{I}:  (\mathcal{P}([-C,C]), \mathcal{W}) \to \mathbb{R} ~, ~\nu \mapsto N_\nu(E) ~\mbox{,}
\end{align}
is continuous at all measures $\nu \in \mathcal{P}([-C,C])$ for which $E$ is a point of continuity (i.e. for which the IDS $N_\nu(E)$ is continuous at $E$).
\item[(ii.)] The modulus of continuity of the map  \eqref{eq_mainmap-cont1} for the DOSm in part (i) is quantified by the following: There exist $\rho_2>0$ and a degree of regularity $M_2 \in \mathbb{N}$, only depending on $d$, such that for each given $C>0$, there is a finite constant $C_{3}>0$, depending on $d$ and $C$, so that for all measures $\nu_1, \nu_2 \in \mathcal{P}([-C,C])$ with $d_w(\nu_1, \nu_2) < \rho_2$ and all functions $f \in \mathcal{C}_c^{M_2}(\mathbb{R})$ with $\mathrm{supp}(f) \subseteq [-r,r]$ and $r\geq 1$, one has
\begin{equation}
| n_{\nu_1}^{(\infty)}(f) - n_{\nu_2}^{(\infty)} (f)| \leq C_{3} r^{2 (k_v + 1)} \Vert f \Vert_{\mathcal{C}^{M_2}} ~\cdot d_w(\nu_1,\nu_2)^{\frac{1}{1+d}} ~\mbox{.}
\end{equation}
\item[(iii.)] Concerning the modulus of continuity of the map \eqref{eq_mainmap-IDScont1} for the IDS in part (i), if $d=1,2,3$, there exists $0<\gamma_d$, only depending on $d$, such that for each $E_0 \in \mathbb{R}$, there is a constant $C_{4;E_0}>0$, depending on $E_0$, $d$, and $C$, such that for every $E \in \mathbb{R}$ with $E \leq E_0$, one has:
\begin{equation} \label{eq_main_contiIDS}
| N_{\nu_1} (E) - N_{\nu_2}(E) | \leq \frac{C_{4;E_0}}{ \left[  \log   \left(   \frac{1}{ d_w(\nu_1,\nu_2)}                   \right)          \right]^{\gamma_d }     } .
\end{equation}
\end{itemize}
\end{theorem}

\begin{remark}
\begin{itemize}
\item[(i)] In view of the statement about the IDS in (\ref{eq_mainmap-IDScont1}) of item (i), we note that we are using the general fact that weak convergence of measures implies point-wise convergence of the respective cumulative distribution functions at each point of continuity of the limiting measure. While for the discrete models described in hypothesis {\bf[{Disc}]} the IDS has long been known to be everywhere continuous in the energy \cite{craig-simon1, DelyonSouillard_1984}, for the continuum models described in hypothesis
{\bf[{Cont}]}, continuity of the IDS in the energy is known in general only for dimensions $d=1,2,3$ and the case of bounded potentials \cite{bourgain-klein1}, see \eqref{eq_IDScont}. For a brief review of available results about the continuity of the IDS in the energy, we refer the reader to e.g. \cite{hislop_marx_1}, section 1.2. therein.
\item[(ii)] The exponent $\gamma_d$ for the fractional $\log$-H\"older dependence in (\ref{eq_main_contiIDS}) can be determined based on the known modulus of continuity of the IDS in the energy for $d=1,2,3$, proven in  \cite{bourgain-klein1}: there it is shown that one can take $\gamma_1=1$, $\gamma_2 = 1/4$, and $\gamma_3 = 1/8$, see Theorem 1.1 therein.
\item[(iii)] As in the discrete case, our proof shows that one can take $M_2 = d+3$  and the radius $\rho_2 = \left( \frac{2}{3} \right)^{1 + d}$.
\end{itemize}
\end{remark}

Since our first paper \cite{hislop_marx_1} was posted, we received comments from I.\ Kachkovskiy \cite{kachkovskiy1}  in which he indicated a different proof for models on $\Z^d$ with single-site probability measures of compact support that gives Lipschitz continuity of the DOSm.
Independently,  M.\ Shamis  communicated another proof for the models considered here using different methods \cite{shamis1} also giving Lipschitz continuity for the DOSm with respect to the single-site probability measure. Shamis uses the Kantorovich-Rubinstein metric instead of the bounded Lipschitz metric \eqref{eq:LipMetric1}. These two metrics are comparable for measures of compact support. Shamis proved that the DOSm $n_{\nu}^{(\infty)}$ satisfies the optimal bound
$$
d_{KR} (n_{\nu_1}^{(\infty)}, n_{\nu_2}^{\infty} ) \leq d_{KR}  ( \nu_1, \nu_2) .
$$
Both authors use the restrictions of the random Schr\"odinger operators to finite volumes, whereas the present work uses the infinite volume operators (one advantage is the latter allows treatment of the Bethe lattice). Shamis then uses the Ky Fan inequality for the eigenvalues whereas Kachkovskiy uses a Hilbert-Schmidt norm inequality for Lipschitz functions of operators
(see \cite{GPS}):
$$
\| f(A) - f(B) \|_2 \leq \|f \|_{Lip} \| A - B\|_2,
$$
for self-adjoint operators $A, B$ and Lipschitz functions $f$.
We thank these authors for sharing their results with us.


\vspace{.1in}
\noindent
{\bf Acknowledgements} We thank S.\ Jitomirskaya for several discussions and for the invitation to work together at UCI. We also thank I.\ Kachkovskiy, M.\ Shamis, and A.\ Skripka for several discussions on topics related to this work.


\section{Step 1 - Finite range reduction with resolvents} \label{sec_step1}
\setcounter{equation}{0}

In \cite{hislop_marx_1}, we studied the lattice model \eqref{eq_discreteschrodop} satisfying {\bf [Disc]} (the same as [H1] in \cite{hislop_marx_1}) under the additional hypothesis that the single-site probability measure $\nu$ has compact support. As pointed out in step 1 of section \ref{sec_intro_context} (``{\em{the finite range reduction}}''), this implied that $\omega \in \Omega \mapsto {\rm Tr} P_0 f(H_\omega) P_0$, for any polynomial $f$ depends on at most finitely many random variables. 

In this section, we modify step 1 of section \ref{sec_intro_context} using resolvents, a method that is flexible enough to allow us to treat 1) lattice models with non-compactly supported single-site measures having finite first moments, and 2) random
 Schr\"odinger operators on $\R^d$. In view of the final continuity result for continuum operators, Theorem \ref{thm:main-cont1}, we point out that, although this modified ``finite range reduction,'' Proposition \ref{eq:finite-range-cont1}, also holds for continuum random Schr\"odinger operators with non-compactly supported single-site measures, the estimates for the Lipschitz property for continuum operators (see step 2 of section \ref{sec_intro_context}) given in Proposition \ref{prop_lipschitz_traceclass_modif3} require compactly-supported single-site probability measures.


\subsection{Finite range reduction for lattice models} \label{subsec:step1-discrete}

The main result of this section is a replacement of \cite[Lemma 2.3]{hislop_marx_1} suitable for the lattice model \eqref{eq_discreteschrodop} satisfying {\bf [Disc]} with non-compactly supported single-site measures $\nu$.
\vspace{.1in}

\noindent
{\bf Assumption 1.} We let $\nu \in \mathcal{P}(\mathbb{R})$ denote a fixed single-site probability measure {\em{with a finite first moment}}, $\mu_1[\nu] < \infty$.

\vspace{.1in}

Given $L \in \mathbb{N}$, we decompose the operator described in \eqref{eq_discreteschrodop} satisfying {\bf [Disc]} according to
\begin{equation} \label{eq_finiterange_decomp}
H_\omega = \left( - \Delta + \sum_{j \in \mathcal{J}: ~\Vert j \Vert_\infty \leq K L} \omega_j P_j  \right) + \sum_{j \in \mathcal{J}: ~\Vert j \Vert_\infty > LK} \omega_j P_j =: H_{\omega; L}^{(0)} + H_{\omega; L}^{(1)} ~\mbox{.}
\end{equation}
Observe that the potential in $H_{\omega; L}^{(0)}$ depends on only {\em{finitely-many}} random variables, $\omega_j$ with $j \in \mathcal{J}$ and $\Vert j \Vert_\infty \leq K L$.

To formulate the finite range reduction for the lattice models in \eqref{eq_discreteschrodop}, given a function $f \in \mathcal{C}_c^\infty(\mathbb{R})$ and $\beta \in \mathbb{N}$, we set
\begin{equation}\label{eq:r-norm1}
\Vert f \Vert_\beta := \sum_{j=0}^{\beta} \int_{-\infty}^\infty \vert f^{(j)} (x) \vert \langle x \rangle^{j-1} ~\mathrm{d} x ~\mbox{,}
\end{equation}
where, as usual, we set $\langle x \rangle := \sqrt{1 + \vert x \vert^2}$.

For $L \in \mathbb{N}$, we define the finite product probability measure $\nu^{(L)}$ by
 $$
\nu^{(L)}  := \bigotimes_{\{j \in \mathcal{J} ~|~ \|j\|_\infty \leq K L\}} ~ \nu .
$$
Then, we claim:

\begin{prop} \label{thm_finiterangereduction}
There exists a constant $c_1$, only depending on $d$ and $N$, such that for every single-site measure $\nu \in \mathcal{P}(\mathbb{R})$, satisfying Assumption 1, and every $L \in \mathbb{N}$, one has
\begin{equation}
n_\nu^{(\infty)}(f) = \dfrac{1}{N} \mathbb{E}_{\nu^{(L)}} [ \mathrm{Tr}(P_0 f(H_{\omega; L}^{(0)}) P_0)] + R_f[L; \nu] ~\mbox{,}
\end{equation}
where the remainder term is controlled by
\begin{equation}\label{eq_thm_finite1}
\vert R_f[L; \nu] \vert \leq c_1 \dfrac{\mu_1[\nu]}{N L}  \Vert f \Vert_{3+d} ~\mbox{.}
\end{equation}
for all $f \in \mathcal{C}_c^\beta(\mathbb{R})$ with $\beta \geq 3 + d$.
\end{prop}

The following proof takes advantage of almost analytic extensions and the Helffer-Sj\"ostrand functional calculus; for convenience of the reader, we briefly review some essential facts in the appendix, section \ref{sec:appendix-helffer-sj1}.
\begin{proof}
1. 
For each fixed $z \in \mathbb{C}\setminus \mathbb{R}$, the second resolvent identity applied to the decomposition in (\ref{eq_finiterange_decomp}) gives
\begin{eqnarray}
F(\omega) & : = & \mathrm{Tr} (P_0 (H_{\omega} - z)^{-1} P_0)  \\
& = & \mathrm{Tr} (P_0 (H_{\omega; L}^{(0)} - z)^{-1} P_0) - \sum_{j \in \mathcal{J}: ~\Vert j \Vert_\infty > K L} \omega_j \mathrm{Tr} (P_0 (H_\omega - z)^{-1} P_j (H_{\omega; L}^{(0)} - z)^{-1} P_0  ) \nonumber \\
& =: & F_{\omega; L}^{(0)}(z) + F_{\omega; L}^{(1)}(z) ~\mbox{.}
\end{eqnarray}
Notice that the random variable $F_{\omega; L}^{(0)}(z)$ only depends on the {\em{finitely}} many components $\omega_j$ of $\omega$ with $\Vert j \Vert_\infty \leq K L$.

\noindent
2. We recall the Combes-Thomas estimate in the trace class. For a textbook presentation of the Combes-Thomas estimate for discrete Schr\"odinger operators on graphs, we refer to Aizenman-Warzel \cite[section 10.3]{AizenmanWarzel1}).
Extensions to bounds in higher trace norms may be found in \cite{saxton2014} and in \cite{shen}.
In the case of our model [H1], there exist constants $c_2, c_3, c_4, c_5 > 0$, only depending on the dimension $d$ and the rank $N$, such that for all $z \in \mathbb{C} \setminus \mathbb{R}$, one has
\beq\label{eq:combes-thomas-tr1}
\vert \mathrm{Tr} (P_0 (H_\omega - z)^{-1} P_j )  \vert \leq \|P_0 (H_\omega -  z)^{-1} P_j\|_1 \leq    \dfrac{c_2}{\vert \mathrm{Im} z \vert} \mathrm{e}^{-c_3 \vert \mathrm{Im} z \vert \vert j \vert} ~\mbox{, }
\eeq
and we also have the bound in the operator norm,
\beq\label{eq:combes-thomas-op1}
\vert \mathrm{Tr} (P_0 (H_{\omega; L}^{(0)} - z)^{-1} P_j ) \vert \leq c_4  \| P_0 (H_{\omega;L}^{(0)} - z)^{-1} P_j \| \leq \dfrac{c_5 }{\vert \mathrm{Im} z \vert} \mathrm{e}^{-c_3 \vert \mathrm{Im} z \vert \vert j \vert}.
\eeq
Using these bounds and the standard trace inequality $\| AB \|_1 \leq \|A \|_1 \|B\|$, for $A$ trace class and $B$ bounded, we thus find
\bea\label{eq:tr-est11}
\vert F_{\omega; L}^{(1)}(z) \vert & \leq  & \sum_{j \in \mathcal{J}: ~\Vert j \Vert_\infty > K L} |\omega_j|
|\mathrm{Tr} (P_0 (H_\omega - z)^{-1} P_j (H_{\omega; L}^{(0)} - z)^{-1} P_0  ) | \nonumber \\
 & \leq &  \sum_{j \in \mathcal{J}: ~\Vert j \Vert_\infty > K L} |\omega_j|  \| P_0 (H_\omega - z)^{-1} P_j \|_1 ~  \| P_j (H_{\omega; L}^{(0)} - z)^{-1} P_0 \| \nonumber \\
 &\leq & \dfrac{c_2^2 }{ \vert \mathrm{Im} z \vert^2 } \sum_{j \in \mathcal{J}: ~\Vert j \Vert_\infty > K L} \vert \omega_j \vert ~\mathrm{e}^{-2 c_3 \vert \mathrm{Im} z \vert \vert j \vert}.
\eea
Hence, averaging $F_{\omega; L}^{(1)}(z)$ with respect to the product measure $\nu^{(\infty)}$ yields
\begin{align} \label{eq_finiterange_resolvent_keybound}
\mathbb{E}_{\nu^{(\infty)}} \left[ \vert F_{\omega; L}^{(1)}(z) \vert \right] & \leq \dfrac{c_2^2  \mu_1[\nu] }{\vert \mathrm{Im} z \vert^2 }  \sum_{j \in \mathcal{J}: ~\Vert j \Vert_\infty > K L} \mathrm{e}^{-2 c_3 \vert \mathrm{Im} z \vert \vert j \vert} \leq \dfrac{c_2^2  \mu_1[\nu] }{\vert \mathrm{Im} z \vert^2 } \left[  \int_{K L}^\infty \mathrm{e}^{-2 c_3  \vert \mathrm{Im} z \vert s} ~ s^{d-1} ~\mathrm{d}s   \right] \nonumber \\
& \leq \dfrac{c_4  \mu_1[\nu] }{ \vert \mathrm{Im} z \vert^{2 + d} } \mathrm{e}^{-2d c_3 \vert \mathrm{Im} z \vert K L},
\end{align}
where $c_4$ depends on $d$ and the other constants.

\noindent
3. Let $f \in \mathcal{C}_c^\beta(\mathbb{R})$ be fixed with $\beta \geq  3 + d$. Choosing the degree of the almost analytic extension to equal $2 + d$, the Helffer-Sj\"ostrand formula (\ref{eq_helffersjostrand}) applied to the operators in the decomposition (\ref{eq_finiterange_decomp}) of $H_\omega$ and the bound in (\ref{eq_finiterange_resolvent_keybound}) yields
\begin{equation} \label{eq_expectationdecomp}
\dfrac{1}{N} \mathbb{E}_{\nu^{(\infty)}} [ \mathrm{Tr}(P_0 f(H_\omega) P_0) ] = \dfrac{1}{N}\mathbb{E}_{\nu^{(\infty)}} [ \mathrm{Tr}(P_0 f(H_{\omega; L}^{(0)}) P_0)] + R_f[L; \nu] ~\mbox{,}
\end{equation}
where
\begin{align}\label{eq_expectationdecomp2}
R_f[L; \nu] & = \dfrac{1}{N \pi} \int\int_\mathbb{C} \partial_{\overline{z}} \widetilde{f}(x,y) \cdot \mathbb{E}_{\nu^{(\infty)}} [F_{\omega; L}^{(1)}(z)] ~\mathrm{d} x ~\mathrm{d} y ~\mbox{.}
\end{align}

\noindent
4. To bound the remainder term $R_f[L; \nu]$ in (\ref{eq_expectationdecomp})--\eqref{eq_expectationdecomp2}, we use the definition of the almost analytic extension in (\ref{eq_almostanalyticextension}). Using the definition of the rescaled bump in (\ref{eq_rescaledBump}), straight-forward estimation shows that for all $z = x + i y \in \mathbb{C}$, one has
\begin{equation}
\vert \sigma_x + i \sigma_y \vert \leq \dfrac{3 \Vert \tau^\prime \Vert_\infty}{\langle x \rangle} ~\chi_U(x,y) ~\mbox{,}
\end{equation}
where we set
\begin{equation}
U:= \{ z \in \mathbb{R} ~:~ \langle x \rangle < \vert y \vert < 2 \langle x \rangle\} ~\mbox{.}
\end{equation}
Hence, letting
\begin{equation}
V:= \{ z \in \mathbb{C} ~:~  \vert y \vert < 2 \langle x \rangle\} ~\mbox{,}
\end{equation}
\eqref{eq_almostanalytic} and the bound \eqref{eq_finiterange_resolvent_keybound} result in
\bea\label{eq_almostanalytic2}
\lefteqn{| \partial_{\overline{z}} \widetilde{f}(x,y)| \E_{\nu^{(\infty)}} [ |F_{\omega;L}^{(1)} (z)|]}
  \nonumber \\
  &\leq  &
c_4 \mu_1[\nu] \left\{ \sum_{n=0}^{2+d} |f^{(n)}(x)| \langle x \rangle^{n-3-d} \chi_U(x,y) +  | f^{(d+3)}(x)| \chi_V(x,y) \right\}  \mathrm{e}^{-2d c_3 \vert y \vert K L}. \nonumber \\
  &  &
\eea
Carrying out first the integration with respect to $y$ in the double integral (\ref{eq_expectationdecomp2}), we conclude the estimate for the remainder term as in \eqref{eq_thm_finite1}.

\noindent
5. Finally, since $H_{\omega; L}^{(0)}$ only depends on components $\omega_j$ of $\omega$ with $\Vert j \Vert_\infty \leq K L$, we can replace the expectation with respect to the {\em{infinite}} product measure $\nu^{(\infty)}$ on the right hand side of (\ref{eq_expectationdecomp}) by an expectation with respect to the {\em{finite}} product measure $\nu^{(L)}$. Consequently, we may write the first term on the right in \eqref{eq_expectationdecomp} as
\begin{equation}
\mathbb{E}_{\nu^{(\infty)}} [\mathrm{Tr}(P_0 f(H_{\omega; L}^{(0)}) P_0)] = \mathbb{E}_{\nu^{(L)}} [ \mathrm{Tr}(P_0 f(H_{\omega; L}^{(0)}) P_0) ] ~\mbox{.}
\end{equation}
This completes the proof of Proposition \ref{thm_finiterangereduction}.
\end{proof}

\vspace{.1in}
In summary, we achieve the finite-range reduction outlined in step 1 of section \ref{sec_intro} for the case that the probability measure $\nu$ is not necessarily compactly supported but only has finite first moment.


\subsection{Finite range reduction for continuum models} \label{subsec:step1-continuum}

We consider the continuum model defined in hypothesis { \bf[{Cont}]}.  For the continuum model  (\ref{eq:random-pot-cont1}), the analog of the projection $P_j$ in \eqref{eq_finiterange_decomp} is multiplication by the function $\phi_j(x) = \phi(x-j)$,
for $j \in \mathbb{Z}^d$.
In contrast to the discrete case, this multiplication operator, and likewise the multiplication operator $\chi_0$ appearing in the definition of the DOSm (\ref{eq_DOSm_cont}), is  no longer finite rank, nor is it trace class relative $H_\omega$ for dimensions $d \geq 2$. We will overcome this by ``regularizing the operator,'' i.e. using the fact that for $m > \frac{d}{2}$, the operator $\chi_0 R_\omega(-i)^m$ is trace class; here, we denote $R_\omega(z) := (H_\omega - z)^{-1}$ and $R_0(z)=(-\Delta - z)^{-1}$, for $z \in \mathbb{C}\setminus \mathbb{R}$. 
Applying this, we modify the operator $\chi_0 f(H_\omega)$ to $\chi_0 F_m (H_\omega)$, where $F_m(H_\omega) := R_\omega(-i)^{-m} f(H_\omega)$ is bounded for every $f \in \mathcal{C}_c(\mathbb{R})$.
We assume that the single-site probability measure is compactly supported in the bounded interval $[-C, C]$.

%
%

We begin the proof in a manner similar to that in section \ref{subsec:step1-discrete}. Given $L \in \mathbb{N}$, we decompose the operator in (\ref{eq_contschrodop}) according to
\begin{equation} \label{eq_finiterange_decomp-cont1}
H_\omega = \left( - \Delta + \sum_{j \in \mathbb{Z}^d: ~\Vert j \Vert_\infty \leq L} \omega_j \phi_j(x)  \right) + \sum_{j \in \mathbb{Z}^d: ~\Vert j \Vert_\infty > L} \omega_j \phi_j(x) =: H_{\omega; L}^{(0)} + H_{\omega; L}^{(1)} ~\mbox{.}
\end{equation}
As before, because of the truncation, $H_{\omega; L}^{(0)}$ depends only on finitely-many random variables associated with the lattice points $j \in \Z^d$ satisfying $\|j\|_\infty \leq L$. Associated with the decomposition in (\ref{eq_finiterange_decomp-cont1}), for $z \in \mathbb{C} \setminus \mathbb{R}$, we let $R_{\omega;L}^{(l)}(z):=( H_{\omega; L}^{(l)} - z)^{-1}$, for $j=0,1$.

We will use the following well known result, see e.g. Corollary 4.8 in \cite{simon_traceclassideals_book} (which also includes a discussion of its history, including references). Recall that a function $f \in L^2_\delta(\mathbb{R}^d)$, for some $\delta > 0$, if $\langle x \rangle^\delta f(x) \in L^2(\mathbb{R}^d)$, in which case one sets $\Vert f \Vert_{L^2_\delta} := \Vert \langle x \rangle^\delta f(x) \Vert_{L^2}$.
\begin{theorem} \label{thm_traceclassop}
Let $f,g \in L^2_\delta(\mathbb{R}^d)$ for some $\delta > d/2$. Then, $f(x) g(-i \nabla) \in \mathcal{S}_1(L^2(\mathbb{R}^d))$ and, for some $C_S = C_S(\delta, d) >0$, one has
\begin{equation}
\Vert f(x) g(-i \nabla) \Vert_{\mathcal{S}_1} \leq C_S \Vert f \Vert_{L^2_\delta} \Vert g \Vert_{L^2_\delta} ~\mbox{.}
\end{equation}
\end{theorem}
Theorem \ref{thm_traceclassop} in particular implies that for $m > d/2$, the operator $\chi_0 R_0(-i)^m$ is trace class.

To formulate the finite range reduction for the continuum models, we introduce the norms
\beq\label{eq:r-norm2}
\| f \|_{k;l} := \sum_{j=0}^k \int_\R \langle s \rangle^{l +j - 1} |f^{(j)}(s)| ~\mathrm{d}s ~\mbox{, for $l, k \in \mathbb{N} \cup \{0\}$,}
\eeq
and $f \in \mathcal{C}_c(\mathbb{R})$. We note that $\| f \|_{k;0} = \| f \|_k$ defined in \eqref{eq:r-norm1}.

\begin{prop} \label{thm_finiterangereduction-cont1}
Let $m > \frac{d}{2}$, and assume that the single-site potential $\phi$ in \eqref{eq:random-pot-cont1} satisfies $\phi \in C_c^{2(m-1)} (\R^d; \R)$ and the single-site probability measures belong to $\mathcal{P}([-C,C])$ . There exists a constant $c_4 = c_4(d, \phi, C)$, depending only on $d$, $C$, and $\| D^\alpha \phi \|_\infty$, for $|\alpha| \leq 2(m-1)$, such that for every single-site measure $\nu \in \mathcal{P}([-C,C])$, for every $L \in \mathbb{N}$, and for every $f \in \mathcal{C}_c^\beta(\mathbb{R})$, with $\beta \geq d + 3$, one has
\begin{equation}
n_\nu^{(\infty)}(f) = \mathbb{E}_{\nu^{(L)}} [ \mathrm{Tr}(\chi_0 f(H_{\omega; L}^{(0)}) \chi_0)] + R_f[L; \nu] ~\mbox{,}
\end{equation}
where
\begin{equation}
\vert R_f[L; \nu] \vert \leq  \dfrac{c_4}{L}  \Vert f \Vert_{3+d; m} ~\mbox{.}
\end{equation}
\end{prop}

\begin{proof}
1. We choose $m > \frac{d}{2}$ and let
\begin{equation}
F_m(s) := (s + i)^m f(s) ~\mbox{,}
\end{equation}
so that
\beq\label{eq:reg-resolv1}
{\rm Tr} \chi_0 f(H_\omega) \chi_0 = {\rm Tr} \chi_0 R_\omega(-i)^m  F_m(H_\omega) \chi_0 ~\mbox{.}
\eeq
It will follow from the argument below that the trace on the right hand side of (\ref{eq:reg-resolv1}) is indeed finite.
Referring to the decomposition in (\ref{eq_finiterange_decomp-cont1}), we wish to replace $f(H_\omega)$ by $f(H^{(0)}_{\omega,L})$. To do so, we first express the right hand side of \eqref{eq:reg-resolv1} using the Helffer-Sj\"ostrand formula:
\beq\label{eq:hl-sj1}
{\rm Tr} \chi_0 R_\omega (- i)^m  F_m(H_\omega) \chi_0 = \frac{1}{\pi} \int_\C
\partial_{\overline{z}} \widetilde{F}_m(x,y) ~ {\rm Tr} [ \chi_0 R_\omega (- i)^m  R_\omega(z) \chi_0 ] ~\mathrm{d}x \mathrm{d}y ~\mathrm{.}
\eeq
Next, we rewrite $\chi_0 R_\omega (-i)^m  R_\omega (z) \chi_0$ using the second resolvent identity $R_\omega (z) = R_{\omega,L}^{(0)}(z) -   R_{\omega,L}^{(0)}(z) H_{\omega,L}^{(1)}(z) R_\omega(z)$ which yields
\bea\label{eq:reg-resolv3}
\chi_0 R_\omega (-i)^m  R_\omega (z) \chi_0 & = & \chi_0 R^{(0)}_{\omega,L}(- i)^m R^{(0)}_{\omega,L}(z) \chi_0  \nonumber \\
 & & - \chi_0 R^{(0)}_{\omega,L}(- i)^m  R_{\omega,L}^{(0)}(z) H_{\omega; L}^{(1)} R_\omega(z) \chi_0                                   \nonumber \\
  & & + \chi_0 \left( R_\omega (- i)^m -  R^{(0)}_{\omega,L}(- i)^m \right) R_\omega(z) \chi_0
 \eea

\noindent
2. Substituting \eqref{eq:reg-resolv3} into the integral of the Helffer-Sj\"ostrand formula \eqref{eq:hl-sj1}, we thus obtain
\beq\label{eq:finite-range-cont1}
\E_{\nu^{(\infty)}} \{ \mathrm{Tr} \chi_0 f(H_\omega ) \chi_0 \} = \E_{\nu^{(\infty)}} \{ \mathrm{Tr} \chi_0 f(H^{(0)}_{\omega, L} ) \chi_0 \} + R_{1,f}[L,\nu] + R_{2,f} [L,\nu]
\eeq
where
\beq\label{eq:finite-range-cont-error1}
R_{1,f}[L,\nu] := -\frac{1}{\pi} \E_{\nu^{(\infty)}} \int_{\C} \partial_{\overline{z}} \widetilde{F}_m (x,y) ~ \mathrm{Tr} \{ \chi_0 (R^{(0)}_{\omega,L}(-i))^m R^{(0)}_{\omega,L} (z) H_{\omega; L}^{(1)} R_\omega(z) \chi_0 \} ~\mathrm{d}x \mathrm{d}y ~\mbox{,}
\eeq
and
\beq\label{eq:finite-range-cont-error2}
R_{2,f}[L,\nu] := -\frac{1}{\pi} \sum_{\ell =1}^m  \E_{\nu^{(\infty)}} \int_{\C} \partial_{\overline{z}} \widetilde{F}_m (x,y) ~ \mathrm{Tr} \{ \chi_0 (R^{(0)}_{\omega,L} (-i))^\ell H_{\omega; L}^{(1)} (R_{\omega}(-i))^{m+1 -\ell} R_\omega (z) \chi_0 \} ~\mathrm{d}x \mathrm{d}y ~\mbox{.}
\eeq

\noindent
3. We will estimate the remainder terms \eqref{eq:finite-range-cont-error1} and \eqref{eq:finite-range-cont-error2} using
the Combes-Thomas estimate for continuum Schr\"odinger operators, see also, for example, \cite{barbaroux1}.
Finally, to control the dependence on the random variable we will employ Lemma \ref{lemma:rel-resolv-bd1}, which we prove at the end of this section.

For the term $R_{1,f}[L, \nu]$ in \eqref{eq:finite-range-cont-error1}, we have
\bea\label{eq:finite-range-cont-error1.2}
|R_{1,f}[L,\nu]| &  \leq & \frac{1}{\pi} \E_{\nu^{(\infty)}} \left\{ \int_{\C} \| \chi_0 (R^{(0)}_{\omega,L}(-i))^m \|_1 |\partial_{\overline{z}} \widetilde{F}_m (x,y)| ~ \| R^{(0)}_{\omega,L} (z)\| ~ \| H_{\omega; L}^{(1)} R_{\omega}(z) \chi_0 \| \right\} ~\mathrm{d}x \mathrm{d}y \nonumber \\
 & \leq & \frac{1}{\pi} \int_{\C} \frac{1}{\vert y \vert} |\partial_{\overline{z}}  \widetilde{F}_m (x,y)| ~ \E_{\nu^{(\infty)}} \left\{
 \| \chi_0 (R^{(0)}_{\omega,L}(-i))^m \|_1~ \| H_{\omega; L}^{(1)} R^{(0)}_{\omega}(z) \chi_0 \| \right\} ~\mathrm{d}x \mathrm{d}y ~\mbox{.} \nonumber \\
  & &
\eea
To control the random variable on the last line in \eqref{eq:finite-range-cont-error1.2}, we first note that by Lemma \ref{lemma:rel-resolv-bd1} we have
\beq\label{eq:rel-resolv-bd3}
 \| \chi_0 (R^{(0)}_{\omega,L}(-i))^m \|_1 \leq C_m(\omega_L) \| \chi_0 R_0(-i))^m \|_1,
\eeq
 for a polynomial $C_m(\omega_L)$ of degree $m$ in the random variables $\omega_L := \{ \omega_n ~| n \in \Z^d , \| n \|_\infty \leq L \}$.

\noindent
4. For the second factor of the expectation in (\ref{eq:finite-range-cont-error1.2}), we use the Combes-Thomas estimate in the operator norm. Denote by $\chi_n(x) = \chi_0(x - n)$ the characteristic function of the unit cube centered at $n \in \mathbb{Z}^d$. Then, by the definition of $H_{\omega; L}^{(1)}$ in \eqref{eq_finiterange_decomp-cont1}, the fact that $\phi_n = \phi_n \cdot \chi_n \leq \chi_n$ (see hyothesis { \bf[{Cont}]}), and the Combes-Thomas estimate \eqref{eq:combes-thomas-op1}, we have
\bea\label{eq:appl-ct-cont1}
  \| H_{\omega; L}^{(1)} R_{\omega}(z) \chi_0 \|   & \leq & \sum_{n \in \Z^d : \|n\|_\infty \geq L} ~ |\omega_n|  \|\chi_{n}  R_{\omega}(z) \chi_0 \|    \nonumber \\
 & \leq &  \dfrac{c_5}{\vert \mathrm{Im} z \vert}  \sum_{n \in \Z^d : \|n\|_\infty \geq L} ~ |\omega_n|   \mathrm{e}^{- c_6 \vert n \vert \vert \mathrm{Im} z \vert } ~\mbox{.}
\eea
Combining \eqref{eq:rel-resolv-bd3} and \eqref{eq:appl-ct-cont1}, we thus find
\bea\label{eq:rel-resolv-bd4}
\lefteqn{\E_{\nu^{(\infty)}} \left\{
 \| \chi_0 (R^{(0)}_{\omega,L}(-i))^m \|_1~ \| H_{\omega; L}^{(1)} R^{(0)}_{\omega,L}(z) \chi_0 \| \right\} } \nonumber \\
 & \leq  & \dfrac{c_5}{\vert \mathrm{Im} z \vert} \| \chi_0 R_0(-i))^m \|_1 \sum_{n \in \Z^d : \|n\|_\infty \geq L} \E_{\nu^{(\infty)}} \left\{ C_m(\omega_L) |\omega_n| \right\}
 \mathrm{e}^{- c_6 \vert n \vert \vert \mathrm{Im} z \vert } ~\mbox{.}
 \eea
We note that here we use that our choice $m > d/2$ ensures that $\| \chi_0 R_0(-i))^m\Vert_1$ is finite as a consequence of Theorem \ref{thm_traceclassop}.
Finally, taking the degree of the almost analytic extension $\widetilde{F}_m (z)$ to be $d+2$ and using (\ref{eq:finite-range-cont-error1.2}) and (\ref{eq:rel-resolv-bd4}), analogous estimation as in part 4 of the proof of Proposition \ref{thm_finiterangereduction}, implies that for some constant $C_m>0$ (only depending on $m$ and $C$), one has
\bea\label{eq:appl-ct-cont2}
|R_{1,f}[L,\nu]|  & \leq & \dfrac{C_m}{L} \|\chi_0 (R_0(-i))^m \Vert_1 ~ \| f \|_{d+3; m} ~\mbox{,}
\eea
where the norm $ \| f \|_{d+3; m}$ was defined in \eqref{eq:r-norm2}.

\noindent
5. The estimate for $R_{2,f}[L,\nu]$ is similar:
\beq\label{eq:finite-range-cont-error2.2}
|R_{2,f}[L,\nu]| \leq \frac{1}{\pi} \sum_{\ell =1}^m  \int_{\C} \partial_{\overline{z}} | \widetilde{F}_m (z)| ~
\E_{\nu^{(\infty)}} \| \chi_0 (R^{(0)}_{\omega,L} (-i))^\ell H_{\omega; L}^{(1)} (R_{\omega}(-i))^{m+1 -\ell} R_\omega (z) \chi_0 \|_1 ~\mathrm{d}x \mathrm{d}y ~\mbox{.}
\eeq
We expand $H_{\omega; L}^{(1)}$ over the lattice points $\| n\|_\infty > L$, and use the H\"older inequality in the Schatten trace norms with H\"older pairs $(p,q)$ satisfying either  $p, q > 1$ and $\frac{1}{p} + \frac{1}{q} = 1$, or $(p=1, q=0)$, or $(p=0, q=1)$, where $p=0$ and $q=0$ denote the operator norm. The expectation in \eqref{eq:finite-range-cont-error2.2} is
then  bounded above as
\bea
\label{eq:finite-range-cont-error2.31}
\lefteqn{ \E_{\nu^{(\infty)}} \{  \| \chi_0 (R^{(0)}_{\omega,L} (-i))^\ell H_{\omega; L}^{(1)} (R_{\omega}(-i))^{m+1 -\ell}  R_\omega (z) \chi_0 \|_1 \}  }  \nonumber \\
 & \leq  & \sum_{n \in \mathbb{Z}^d: \Vert n \Vert_\infty \geq L} \E_{\nu^{(\infty)}} \{ \vert \omega_n \vert   \| \chi_0 (R^{(0)}_{\omega,L} (-i))^\ell \chi_n \|_p  ~  \|  (R_{\omega}(-i))^{m+1 -\ell} \chi_0 \|_q \cdot  \| R_\omega (z)  \|  \} .
  \nonumber \\
   & &
\eea
In order to extract exponential decay of the Schatten trace-norms of the localized resolvents in \eqref{eq:finite-range-cont-error2.31}, we need Combes-Thomas estimates in the Schatten trace-norms, such as those obtained in \cite{saxton2014, shen}.
Let $\| A \|_p$ denote the norm of $A$ in the $p^{\rm th}$-trace ideal, for $p \geq 1$. Then, it is proven in \cite{saxton2014, shen} that   there exists finite constants $C, c > 0$, depending on $| \Im z| \in \mathbb{C}\setminus \mathbb{R}$ and $p$, such
 that for all $m,n \in \mathbb{}Z^d$ we have
\beq\label{eq:ct-trclass1}
\| \chi_m R_\omega(z)^k \chi_n \|_p \leq C e^{- c \vert m-n \vert}  ~~\mbox{, provided that} ~~k > \frac{d}{2p} ~\mbox{,}
\ee
The same estimate holds if $R_\omega(z)$ if replaced by $R^{(0)}_{\omega,L}(z)$. Notice that we will use (\ref{eq:ct-trclass1}) to control the integrand in (\ref{eq:finite-range-cont-error2.2}), in particular $z=-i$ is {\em{fixed}} whence the $z$-dependence of the constants in (\ref{eq:ct-trclass1}) will not be relevant here.
In order to estimate the factor $\|  (R_{\omega}(-i))^{m+1 -\ell} \chi_0 \|_q $ on the right of the last line of  \eqref{eq:finite-range-cont-error2.31}, we insert the partition of unity $1 = \sum_{j \in \Z^d} \chi_j$ before the resolvent and obtain
\beq\label{eq:second-factor1}
\|  (R_{\omega}(-i))^{m+1 -\ell} \chi_0 \|_q  \leq \sum_{j \in \Z^d} \| \chi_j (R_{\omega}(-i))^{m+1 -\ell} \chi_0 \|_q .
\eeq
For $q \geq 1$ so that $(m + 1 - \ell) > \frac{d}{2q}$, the Combes-Thomas estimate \eqref{eq:ct-trclass1} implies that there are constants $C', C, c > 0$, depending on $z=-i$ and $q$, so that
\bea\label{eq:second-factor2}
\|  (R_{\omega}(-i))^{m+1 -\ell} \chi_0 \|_q & \leq & \sum_{j \in \Z^d} \| \chi_j (R_{\omega}(-i))^{m+1 -\ell} \chi_0 \|_q  \nonumber \\
 & \leq & \sum_{j \in \Z^d} C' e^{ c | j|}  \leq C < \infty.
\eea

\noindent
6. Given these preliminaries, we first consider the endpoints $\ell = 1$ and $\ell = m$ in \eqref{eq:finite-range-cont-error2.31}.  For the $\ell = 1$ term, we take $(p,q) = (0,1)$ so the expectation is bounded above by
\beq\label{eq:finite-range-cont-error2.21}
\E_{\nu^{(\infty)}} \{ | \omega_n| \| \chi_0 (R^{(0)}_{\omega,L} (-i)) \chi_n\| \cdot \|  R_{\omega}(-i))^{m} \chi_0 \|_1 \cdot
\|  R_\omega (z) \|  \}.
\eeq
The factor $ \| \chi_0 R^{(0)}_{\omega,L} (-i) \chi_n\| $ decays exponentially by the Combes-Thomas estimate \eqref{eq:appl-ct-cont1} in the operator norm, and the second factor is bounded by a constant as in \eqref{eq:second-factor2}.
For $\ell = m$, we take $(p,q) = (1,0)$ and bound the first factor in the trace norm $ \| \chi_0 (R^{(0)}_{\omega,L} (-i))^m \chi_n\|_1$  and use the exponential decay \eqref{eq:ct-trclass1} which is valid  since $p=1$ and $m > \dfrac{d}{2}$. We take $q=0$ in the second factor which is bounded by the Combes-Thomas bound in the operator norm ($q=0$) as in \eqref{eq:second-factor2}.

\noindent
7. For any fixed $1 <  \ell < m$ in (\ref{eq:finite-range-cont-error2.2}), we use H\"older's inequality for trace norms 
with indices $p,q > 1$, $p^{-1} + q^{-1} = 1$, satisfying
\begin{equation} \label{eq_holdercombesth}
l > \dfrac{d}{2p} ~\mbox{and } m+1 - l > \dfrac{d}{2q} ~\mbox{.}
\end{equation}
Such a pair $(p,q)$ satisfying \eqref{eq_holdercombesth} exists. As $ q = \frac{p}{p-1}$, we solve the second inequality of \eqref{eq_holdercombesth} for $p$ and obtain
\beq\label{eq:holder1} \frac{1}{p} > \frac{d - 2(m+1-\ell)}{d} .
\eeq
There are two cases:
\begin{enumerate}
\item Case 1: $d \leq 2(m+1 - \ell)$. In this case, the inequality \eqref{eq:holder1} is always satisfied for $p > 1$.  Consequently we take
\beq\label{eq:holder2}
p > {\rm max} ~ \left\{ 1, \frac{d}{2 \ell} \right\} .
\eeq

\item Case 2: $d > 2 (m+ 1 - \ell)$. In this case, inequality \eqref{eq:holder1} implies
\beq\label{eq:holder3}
p < \frac{d}{d - 2 ( m + 1 - \ell)} ,
\eeq
and the left side of \eqref{eq:holder3} is greater than 1. Consequently $p$ satisfies the constraints:
\beq\label{eq:holder4}
 {\rm max} ~ \left\{ 1, \frac{d}{2 \ell} \right\} < p <  \frac{d}{d - 2 ( m + 1 - \ell)} .
\eeq

\end{enumerate}
Since $m \in \N$ and $m > \frac{d}{2}$, if $d \geq 1$ is even, we take $m = \frac{d}{2} + 1$, and for $d \geq 1$ odd, we take $m = \frac{d+1}{2}$.
Case 2 holds for $d > 4$ for the range $2 \leq \ell \leq m - 1$ provided $d>4$ is odd, and for the range $3 \leq \ell \leq m - 1$ provided $d>4$ is even. For these ranges, we choose $p>1$ according to \eqref{eq:holder4}. When $d>4$ is even, the condition of case 1 is $d \leq  d + 2 (2 - \ell)$ and this holds for $\ell = 2$. In these cases, we choose $p>1$ according to \eqref{eq:holder2}. Hence, for $d > 4$, there is always a pair $(p,q)$ satisfying the above conditions.
It remains to consider the cases $d=1,2,3,4$. For $d = 2,3$, we may take $m=2$. Then, as $1 \leq \ell \leq 2$, we are in the endpoint case treated above in part 6 of the proof. Similarly for $d=1$, we take $m=1= \ell$ and this endpoint case is also treated above. For $d=4$, we take $m = 3$ so $\ell = 1,2,3$. The endpoints $\ell = 1,3$ are treated above and case 1 is satisfied for $\ell = 2$.
 Thus, for each $d \geq 1$ and  $1 <  \ell <  m$, we can find a H\"older pair $(p,q)$ satisfying the above conditions.
Consequently, we can use (\ref{eq:ct-trclass1}) and \eqref{eq:second-factor2} to estimate the trace-norm of the operator in the integrand of \eqref{eq:finite-range-cont-error2.2}:
\bea
\label{eq:finite-range-cont-error2.3}
\lefteqn{ \E_{\nu^{(\infty)}} \{  \| \chi_0 (R^{(0)}_{\omega,L} (-i))^\ell H_{\omega; L}^{(1)} (R_{\omega}(-i))^{m+1 -\ell}  R_\omega (z) \chi_0 \|_1 \}  }  \nonumber \\
 & \leq  & \sum_{n \in \mathbb{Z}^d: \Vert n \Vert_\infty \geq L} \E_{\nu^{(\infty)}} \left\{ \vert \omega_n \vert   \| \chi_0 (R^{(0)}_{\omega,L} (-i))^\ell \chi_n \|_p  ~  \| (R_{\omega}(-i))^{m+1 -\ell} \chi_0 \|_q \cdot  \| R_\omega (z) \|  \right\}
  \nonumber \\
  & \leq & \dfrac{C_1}{\vert \mathrm{Im} z \vert} \sum_{n \in \mathbb{Z}^d: \Vert n \Vert_\infty \geq L} \mathrm{e}^{-  c \cdot \vert n \vert} =: \dfrac{C_2 }{ \vert \mathrm{Im} z \vert L} ~\mbox{,}
 \eea
where $C_2$ only depends on $d$, $p$, $q$, and $z=-i$. In particular, the bound in (\ref{eq:finite-range-cont-error2.3}) and the choice of $d+2$ for the degree of the almost analytic extension implies that
\begin{equation} \label{eq:finite-range-cont-error2.4}
\vert R_{2,f}[L, \nu] \vert \leq \dfrac{C_{3}}{L} \Vert f \Vert_{d+3; m} ~\mbox{,}
\end{equation}
where $C_3 > 0$ depends on $d$ and $C$.
Finally, setting $R_f [L, \nu] := R_{1,f}[L, \nu] + R_{2,f}[L, \nu]$, the combination of (\ref{eq:appl-ct-cont2}) and (\ref{eq:finite-range-cont-error2.4}) completes the proof.
\end{proof}

To conclude this section, we prove the following technical lemma used in the proof of Proposition \ref{thm_finiterangereduction-cont1}.
\begin{lemma}\label{lemma:rel-resolv-bd1}
Let $H_0 := - \Delta$. For any integer $m \in \N$, suppose $V \in C^{2(m-1)}_c (\R^d;\R)$. There is a finite constant $C_m(V) > 0$, depending on the norms $\| D^{\alpha} V \|_\infty$, for multi-indices $\alpha \in \N^d$ and $|\alpha| = 0, \ldots, 2(m-1)$, so that
\beq\label{eq:rel-resolv-bd1}
\| (H_{0} + i)^m(H_0 + V +i)^{-m} \| \leq C_m(V) ~\mbox{.}
\eeq
Consequently, applying bound \eqref{eq:rel-resolv-bd1} to $H_{\omega;L}^{(0)}$, there exists a polynomial $C_m(\omega_L;\phi)$ of degree $m$ in the random variables $\omega_L := \{ \omega_n ~| n \in \Z^d, \|n\|_\infty \leq L  \}$ and $\| D^{\alpha} \phi \|$, for $|\alpha|=0,1, \ldots, 2(m-1)$, so that
\beq\label{eq:rel-resolv-bd2}
\|(R_0(-i))^{-m} ~ (R^{(0)}_{\omega,L}(-i))^{m} \| \leq C_m(\omega_L;\phi) ~\mbox{.}
\eeq
\end{lemma}

\begin{proof}
1. The $H^s$-mapping properties of the resolvent $R_0(-i)$ and the fact that $V \in C^{2(m-1)}_c (\R^d;\R)$ imply that the operator on the left in \eqref{eq:rel-resolv-bd1} is bounded. In order to show the dependence of the constant on the right in \eqref{eq:rel-resolv-bd1} on $V$, we write $A := H_0 + V + i$ so that $A^{-1}$ exists and $\| A^{-1} \| \leq 1$. We need to estimate the norm of
\beq\label{eq:expansion1}
(H_{0} + i)^m(H_0 + V +i)^{-m} = ((H_0 +V + i) - V)^m(H_0 + V +i)^{-m} = (A - V)^m A^{-m}.
\eeq
We note that
\beq\label{eq:expansion2}
(A - V)^m = \sum_{j=0}^{m}  p_j\{ A^{m-j}; V^j \},
\eeq
where $p_j \{ A^{m-j}; V^j \}$ is a sum of terms each of which is a product of $m-j$ factors of $A$ and $j$ factors of $V$ in all possible positions. In order to bound any term in the sum ${p}_j\{ A^{m-j}; V^j \} A^{-m}$,
we need to commute the $j$ factors of $V$ to the left through at most $m-j$ factors of $A$. This involves multiple commutators of $A$ with $V$.  For example, one term is
\beq\label{eq:commutator1}
AV^2 = V^2 A + 2 V [A, V] + [ [ A,V],V] .
\eeq
Since $A$ is a second-order operator, each commutator with $V$ lowers the order by one. As $V \in {\mathcal{C}}^{2(m-1)}_c (\R^d;\R)$, the first commutator is
\beq\label{eq:commutator2}
[A, V ] = - \Delta V - 2 \nabla V \cdot \nabla ,
\eeq
a first-order operator, and the commutator contains two derivatives of $V$.
Consequently, commuting $V$ through $m-j$ factors of $A$ will result in $2(m-j)$ derivatives on $V$, for $1 \leq j \leq m-1$.
As seen from \eqref{eq:commutator2}, we also need the standard bounds such as
\bea\label{eq:resolvent-bounds1}
\| \partial_i (H_0 + V + i )^{-1}\|  & \leq &  (2 + \|V \|_\infty )^{\frac{1}{2}}, ~~~i = 1, \ldots, d. \nonumber \\
\sum_{i,j=1}^d \| \partial_i \partial_j (H_0 + V + i )^{-1}\| & \leq & 2 + \|V \|_\infty .
\eea
This establishes \eqref{eq:rel-resolv-bd1}.

\noindent
2. For the application in \eqref{eq:rel-resolv-bd2}, we note that for $V$ of the form $V_\omega |_{\{ x ~|~ \| x \|_\infty \leq L\}}$, with $V_\omega$ defined in (\ref{eq:random-pot-cont1}), hypothesis $\textbf{[Cont]}$ guarantees the required smoothness. Since the function $\phi$ has support in the unit cube, so that $\phi_j \phi_k = 0$ for $j \neq k$, it follows that for any multi-index $\alpha \in \N^d$ nd $k \in \N$, we have
\beq\label{eq:est-V1}
| D^\alpha V_\omega(x)|^k  \leq   \sum_{ \{ n \in \Z^d ~:~ \|n\|_\infty \leq L \}} |\omega_n|^k |D^\alpha \phi_n (x)|^k ,
\eeq
and $|\omega_n| \leq C$. Hence, the constant $C_m(\omega_L;\phi)$ on the right in \eqref{eq:rel-resolv-bd2} depends only on $C$ and $\| D^\alpha \phi \|_\infty$, for $|\alpha| \leq 2(m-1)$.
\end{proof}


\section{Step 2 - Lipschitz property} \label{sec_step2lipschitzprop}
\setcounter{equation}{0}

Theorem \ref{thm_finiterangereduction} reduces the variation of the {\em{infinite}} product measure $\nu^{(\infty)}$ associated with a given single-site measure $\nu$ to changing the measure at {\em{finitely}} many lattice points in a fixed cube, centered about the origin. Carrying out these variations one lattice point at a time, we are led to analyzing the continuity properties of maps of the form
\begin{equation} \label{eq_setup_Lipschitz}
\omega_{j_0} \mapsto \mathrm{Tr}\left(P_0 f(H_{j_0^\perp} + \omega_{j_0} P_{j_0}) P_0\right) ~\mbox{,}
\end{equation}
for a fixed, arbitrary lattice point $j_0 \in \mathcal{J}$ and a function $f \in \mathcal{C}_c(\mathbb{R})$, where
\begin{equation} \label{eq_setup_Lipschitz_1}
H_{{j_0}^\perp} := - \Delta + \sum_{j \in \mathcal{J} ~:~ j \neq j_0} \omega_j P_j ~\mbox{.}
\end{equation}
While (\ref{eq_setup_Lipschitz}) - (\ref{eq_setup_Lipschitz_1}) are written for the discrete model described in {\bf[{Disc}]}, analogous considerations apply in the continuum { \bf[{Cont}]} with obvious modifications (see section \ref{subsec_step2lipschitzprop_conti} for details).

The purpose of this section is to establish the Lipschitz continuity of (\ref{eq_setup_Lipschitz}) for compactly supported Lipschitz test functions $f$. We will first consider the discrete case for which $\nu$ will be allowed to have unbounded support. In section \ref{subsec_step2lipschitzprop_conti}, we consider the the Lipschitz continuity for the continuum models with $d \geq 1$ where $\nu$ is assumed to have compact support. 



\subsection{Lipschitz property for discrete models} \label{sec_step2lipschitzprop_discrete}

The Lipschitz continuity of the map (\ref{eq_setup_Lipschitz}) for the case of discrete operators (\ref{eq_discreteschrodop}) has already been the subject of Proposition 4.1 in \cite{hislop_marx_1}. Proposition \ref{prop_lipschitz_traceclass} below provides both an extension of the latter to the case of noncompactly supported probability measures $\nu$, and serves a useful starting point for our discussion of the Lipschitz property for continuum models, given in section \ref{subsec_step2lipschitzprop_conti}.

To this end, let $T_1, T_2$ be positive, bounded operators and $H_0$ be a (not necessarily bounded) self-adjoint operator on a given Hilbert space $\mathcal{H}$. We consider the one-parameter family
\begin{equation} \label{eq_lipschitz_func_set-up}
H_\lambda := H_0 + \lambda T_1 ~\mbox{, }  \lambda \in \mathbb{R} ~\mbox{.}
\end{equation}
Given a function $f \in \mathcal{C}_c(\mathbb{R})$, we examine the continuity properties of the map
\begin{equation} \label{eq_lipschitz_func}
\mathbb{R} \ni \lambda \mapsto \mathcal{F}_f(\lambda):= \mathrm{Tr}(T_2 f(H_\lambda) T_2) ~\mbox{.}
\end{equation}
We let $\mathcal{S}_1(\mathcal{H})$ denote the trace-class operators, and $\mathcal{S}_2(\mathcal{H})$ denote the Hilbert-Schmidt operators on $\mathcal{H}$. We write $\Vert . \Vert_{\mathcal{S}_j}$, $j=1,2$ for the associated Banach space norms.

\begin{prop} \label{prop_lipschitz_traceclass}
Given the setup described in (\ref{eq_lipschitz_func_set-up}) - (\ref{eq_lipschitz_func}), with $T_1, T_2$ positive, bounded operators and $H_0$ a (not necessarily bounded) self-adjoint operator on a given Hilbert space $\mathcal{H}$. Suppose $T_1 \in \mathcal{S}_1(\mathcal{H})$ and $T_2 \in \mathcal{S}_2(\mathcal{H})$, then for every $f \in Lip_c(\mathbb{R})$, the map $\lambda \mapsto \mathcal{F}_f(\lambda)$ is Lipschitz in $\lambda$, satisfying
\begin{equation} \label{eq_prop_lipschitz_traceclass}
\left\vert \mathcal{F}_f(\lambda_1) -  \mathcal{F}_f(\lambda_2) \right\vert \leq \min\{ \Vert T_2^2 \Vert \Vert T_1 \Vert_{\mathcal{S}_1} , \Vert T_2 \Vert_{\mathcal{S}_2}^2 \Vert T_1 \Vert \} L_f \cdot \vert \lambda_1 - \lambda_2 \vert ~\mbox{,}
\end{equation}
for each $\lambda_1, \lambda_2 \in \mathbb{R}$. Here, $L_f$ is the optimal Lipschitz constant of $f$ as defined in (\ref{eq_lipschitznorm}).
\end{prop}

\begin{proof}
1. We adapt the proof of Proposition 4.1 in \cite{hislop_marx_1} given in Appendix B therein and will establish (\ref{eq_prop_lipschitz_traceclass}) first for $f \in \mathcal{C}_c^\infty(\mathbb{R})$. Using the Helffer-Sj\"ostrand functional calculus, we represent
\begin{equation}  \label{eq_prop_lipschitz_1}
\mathcal{F}_f(\lambda) =  \dfrac{1}{\pi} \int\int_\mathbb{C} \partial_{\overline{z}} \widetilde{f}(x,y) \cdot \mathrm{Tr}( T_2 (H_\lambda - z)^{-1} T_2 ) ~\mathrm{d} x ~\mathrm{d} y ~\mbox{,}
\end{equation}
where $\widetilde{f}$ is a fixed almost analytic extension of $f$ with degree $2$. By the second resolvent identity, one has, for each $z \in \mathbb{C} \setminus \mathbb{R}$
\begin{equation} \label{eq_prop_lipschitz_deriv}
\frac{d}{d \lambda} {\rm tr} (T_2 (H_\lambda -z)^{-1} T_2) = - {\rm Tr} ( T_2 ( H_\lambda - z)^{-1} T_1 (H_\lambda - z)^{-1} T_2 ) ~\mbox{.}
\end{equation}
Because we chose the degree to be $2$, we can differentiate (\ref{eq_prop_lipschitz_1}) under the integral sign, since by (\ref{eq_almostanalytic})
\begin{equation} \label{eq_prop_lipschitz_deriv_just}
\vert \partial_{\overline{z}} \widetilde{f} \vert \vert {\rm Tr} ( T_2 ( H_\lambda - z)^{-1} T_1 (H_\lambda - z)^{-1} T_2 ) \vert \leq \Vert T_2^2 \Vert \Vert T_1 \Vert_{\mathcal{S}_1} \Vert f^{(3)} \Vert_\infty ~\mbox{,}
\end{equation}
for all $0 < \vert \mathrm{Im} z \vert \leq 1$.

\noindent
2. To obtain a more explicit representation of the right hand side of (\ref{eq_prop_lipschitz_deriv}), consider the bilinear functional on $\mathcal{C}_c(\mathbb{R}) \times \mathcal{C}_c(\mathbb{R})$ given by
\begin{equation} \label{eq_prop_lipschitz_bimeasure}
\beta_\lambda(f,g):= {\rm tr} ( T_2 f(H_\lambda) T_1 g(H_\lambda) T_2 ) ~\mbox{, for }   f,g\in \mathcal{C}_c(\mathbb{R}) ~\mbox{.}
\end{equation}
By the spectral theorem for compact operators, we write
\begin{align} \label{eq_prop_lipschitz_compact}
T_j =\sum_{k \in \mathbb{N}} \lambda_k^{(j)} \vert \phi_k^{(j)} \rangle \langle \phi_k^{(j)} \vert ~\mbox{, }
\end{align}
for appropriate orthonormal bases (ONB) $\{ \phi_k^{(j)} \}$ and $j =1,2$. Using spectral measures, $\beta_\lambda(f,g)$ may thus be represented by a (in general, complex) Baire measure $\mu_\lambda$ on $\mathbb{R}^2$, i.e.
\begin{align}
\beta_\lambda(f,g) & = \sum_{k,l \in \mathbb{N}} (\lambda_l^{(2)})^2 \lambda_k^{(1)} \langle \phi_l^{(2)} , f(H_\lambda) \phi_k^{(1)} \rangle \langle \phi_k^{(1)} , g(H_\lambda) \phi_l^{(2)} \rangle \\
& =: \int_{\mathbb{R}^2} f(s) g(t) ~\mathrm{d} \mu_\lambda(s,t)
\end{align}
where
\begin{equation} \label{eq_prop_lipschitz_defnmeas}
\mathrm{d} \mu_\lambda(s,t) = \sum_{k,l \in \mathbb{N}} (\lambda_l^{(2)})^2 \lambda_k^{(1)} \mathrm{d} \mu_{\phi_l^{(2)}; \phi_k^{(1)}}(s) \otimes \mathrm{d} \mu_{\phi_k^{(1)}; \phi_l^{(2)}}(t) ~\mbox{,}
\end{equation}
and
\begin{align} \label{eq_prop_lipschitz_meas_totalvar}
\vert  \mu_\lambda\vert & \leq \sum_{k,l \in \mathbb{N}} \vert \lambda_l^{(2)} \vert^2  \vert \lambda_k^{(1)}\vert  \vert  \mu_{\phi_l^{(2)}; \phi_k^{(1)}} \otimes \mu_{\phi_k^{(1)}; \phi_l^{(2)}} \vert \nonumber \\
& \leq 2 \left( \sum_{l \in \mathbb{N}} \vert \lambda_l^{(2)} \vert^2 \right) \left( \sum_{k \in \mathbb{N}} \vert \lambda_k^{(1)} \vert \right) = 2 \Vert T_2 \Vert_{\mathcal{S}_2}^2 \Vert T_1 \Vert_{\mathcal{S}_1} ~\mbox{.}
\end{align}

Observe that since $T_1 \geq 0$ and $T_2$ is self-adjoint, $\beta_\lambda(f,f) \geq 0$ for every real-valued $f \in \mathcal{C}_c(\mathbb{R})$, thus the measure $\mu$ is non-negative on all cubes $\Delta \times \Delta \subset \mathbb{R}^2$, and hence on all rectangles with rational ratios of their side lengths. Therefore, we conclude that $\mu$ is in fact a non-negative real measure on $\mathbb{R}^2$ with total mass bounded by
\begin{equation} \label{eq_prop_lipschitz_meas_totalmass}
\mu_\lambda(\mathbb{R}^2) \leq \mathrm{Tr} (T_2 T_1 T_2) \leq \min\{ \Vert T_2^2 \Vert \Vert T_1 \Vert_{\mathcal{S}_1} , \Vert T_2 \Vert_{\mathcal{S}_2}^2 \Vert T_1 \Vert \}
\end{equation}


\noindent
3. Combining (\ref{eq_prop_lipschitz_1}) - (\ref{eq_prop_lipschitz_deriv_just}), the Cauchy-Pompieu formula \cite[section IV.8]{gamelin} thus allows to compute:
\begin{align}
\dfrac{\mathrm{d}}{\mathrm{d} \lambda} \mathcal{F}_f(\lambda) & = - \dfrac{1}{\pi} \int_\mathbb{C} \partial_{\overline{z}} \widetilde{f}(x,y) \cdot \left( \int_{\mathbb{R}^2} \dfrac{1}{(s - z)(t-z)} \mathrm{d} \mu_\lambda(s,t) \right) ~\mathrm{d} x \mathrm{d} y   \label{eq_cauchypompie_deriv_1}  \\
 &= - \int_{\R^2} ~\mathrm{d} \mu_\lambda(s,t) \chi_{\{s \neq t\}} \frac{1}{t-s} ~\frac{1}{\pi} \int_{\C} \partial_{\overline{z}} \widetilde{f}(x,y) \left( \frac{1}{s-z} - \frac{1}{t-z} \right)  ~\mathrm{d} x \mathrm{d} y  \nonumber \\
  & - \int_{\R^2} ~\mathrm{d} \mu_\lambda(s,t) \chi_{\{s = t\}}  ~\frac{1}{\pi} \int_{\C}\partial_{\overline{z}} \widetilde{f}(x,y) \frac{1}{(s-z)^2}  ~\mathrm{d} x \mathrm{d} y \label{eq_cauchypompie_deriv} \\
 &= \int_{\R^2} ~\mathrm{d} \mu_\lambda(s,t) \left\{ \chi_{\{s \neq t\}} \frac{f(t) - f(s)}{t-s} + \chi_{\{s = t\}} f^\prime(s) \right\} \\
 & = \int_{\R^2} ~\mathrm{d} \mu_\lambda(s,t) f^\prime(\zeta_{s,t}) ~\mbox{,} \label{eq_prop_lipschitz_diffcomp}
\end{align}
where, as a result of the mean-value theorem, $\zeta_{s,t}$ is some point between $s$ and $t$. We note that our choice of $2$ for the degree of the almost analytic extension of $f$ and (\ref{eq_almostanalytic}) imply that for all points on the real axis, the Cauchy-Pompeiu formula may be differentiated to yield the representation of the derivative of $f$ given in (\ref{eq_cauchypompie_deriv}).

From the bound in (\ref{eq_prop_lipschitz_meas_totalmass}), we thus conclude the claim for all $f \in \mathcal{C}_c^\infty(\mathbb{R})$, since by (\ref{eq_prop_lipschitz_diffcomp}), $\mathcal{F}_f$ is differentiable for each $ \lambda \in \mathbb{R}$ and
\begin{align}
\left\vert \dfrac{\mathrm{d}}{\mathrm{d} \lambda} \mathcal{F}_f(\lambda) \right\vert \leq \min\{ \Vert T_2^2 \Vert \Vert T_1 \Vert_{\mathcal{S}_1} , \Vert T_2 \Vert_{\mathcal{S}_2}^2 \Vert T_1 \Vert \} \cdot L_f ~\mbox{.}
\end{align}

Finally, a simple approximation argument using $\mathcal{C}^\infty$ mollifiers (see e.g.\ item 5 in the proof of Proposition 4.1 in \cite{hislop_marx_1}), allows to extend the result to all $f \in Lip_c(\mathbb{R})$.
\end{proof}

Applying Proposition \ref{prop_lipschitz_traceclass} to the situation described in (\ref{eq_setup_Lipschitz}) - (\ref{eq_setup_Lipschitz_1}) with $H_0 = H_{j^\perp}$, $T_1 = P_j$, and $T_2 = P_0$ (i.e.  $\Vert T_2 \Vert_{\mathcal{S}_2}^2 = \Vert T_1 \Vert_{\mathcal{S}_1} = N$ and $\Vert T_1 \Vert = \Vert T_2^2 \Vert = 1$) yields:
\begin{corollary}[Lipschitz property for discrete random Schr\"odinger operators] \label{coro_lipschitzdiscrete}
For the discrete random Schr\"odinger operators in (\ref{eq_discreteschrodop}), each $j \in \mathbb{Z}^d$, and $f \in Lip_c(\mathbb{R})$, the map in (\ref{eq_setup_Lipschitz}) is Lipschitz with Lipschitz constant bounded above by $N \cdot L_f$.
\end{corollary}
As mentioned at the beginning of section 4 of \cite{hislop_marx_1}, weaker versions of Corollary \ref{coro_lipschitzdiscrete} could be extracted {\em{directly}} from known properties of operator-valued Lipschitz functions; the latter will however yield a dependence on higher order derivatives of $f$ in the upper bound (\ref{eq_prop_lipschitz_traceclass}), which is in essence equivalent to the bound in (\ref{eq_prop_lipschitz_deriv_just}).

\begin{remark} \label{rem_positivemeas_1}
We point out that Corollary \ref{coro_lipschitzdiscrete} improved the Lipschitz constant ($N \cdot L_f$) compared to our earlier result in \cite{hislop_marx_1}, Proposition 4.1 ($2 N^2 \cdot L_f$). This is a consequence of the observation that the measure $\mu_\lambda$ defined in (\ref{eq_prop_lipschitz_defnmeas}) is not merely complex but in fact positive. In view of (\ref{eq_prop_lipschitz_diffcomp}), this allows to replace the bound on the total variation of $\mu_\lambda$ in (\ref{eq_prop_lipschitz_meas_totalvar}) with simply a bound on the total mass, given in (\ref{eq_prop_lipschitz_meas_totalmass}).
\end{remark}


\subsection{Lipschitz property for continuum models} \label{subsec_step2lipschitzprop_conti}

The Lipschitz property for continuum models may be formulated as follows. The random Schr\"odinger operators $H_\omega$ have the form $H_\omega = - \Delta + V_\omega$ where, for $g \in S(\R^d)$, the random potential $V_\omega$
is defined as in \eqref{eq:random-pot-cont1} by
$$
(V_\omega g)(x) := \sum_{j \in \Z^d} \omega_j \phi (x-j) g(x) ~\mbox{.}
$$

As before, we write $\chi_0$ for multiplication by the characteristic function $\chi_{\Lambda_0}$, where $\Lambda_0$ is the unit cube centered at the origin. The Lipschitz property for continuum models \eqref{eq_discreteschrodop} involves the variation of the $j^{th}$-random variable
\begin{equation} \label{eq_setup_Lipschitz-cont1}
\omega_j \mapsto \mathrm{Tr}\left(\chi_0 f(H_{j^\perp} + \omega_j \phi_j) \chi_0\right) ~\mbox{,}
\end{equation}
for an arbitrary, fixed lattice point $j \in \mathbb{Z}^d$, functions $f \in \mathcal{C}_0(\mathbb{R})$, and
\begin{equation} \label{eq_setup_Lipschitz_2}
H_{j^\perp} = - \Delta + \sum_{n \in \mathbb{Z}^d: n \neq j} \omega_n \phi_n ~\mbox{.}
\end{equation}

As in section \ref{subsec:step1-continuum}, for continuum models, we assume that the probability measures $\nu_\alpha, \nu \in \mathcal{P}([-C, C])$, for $0 < C < \infty$. This ensures that the random potential is bounded: $|V_\omega (x)| < M_C < \infty$ for all $\omega$. Consequently, the operator $H_\omega$ is lower semi-bounded and the expressions \eqref{eq:kappas1} are finite.

In order to prove the Lipschitz property for continuum models in all dimensions $d \geq 1$, we will regularize the trace similar to section \ref{subsec:step1-continuum}, i.e. we will rely on the fact that on $L^2(\mathbb{R}^d)$, the operator $\chi_0 (-\Delta + z)^{-m}$ is trace class for all exponents $m > \frac{d}{2}$ and $z \in \mathbb{C} \setminus [0, +\infty)$, with an explicit trace-norm estimate given by Theorem \ref{thm_traceclassop}.

The desired Lipschitz property for continuum Schr\"odinger operators will follow as a corollary to a generalization of Proposition \ref{prop_lipschitz_traceclass}. The main point of the latter is to relax the trace class condition for the operators $T_1$ and $T_2$ in Proposition \ref{prop_lipschitz_traceclass} by assuming that both operators be merely trace class {\em{relative to appropriate powers of the resolvent}} of the one-parameter family.

To formulate it, for $\lambda \in \mathbb{R}$ and $f \in Lip_c(\mathbb{R})$, let $H_\lambda$ and $\mathcal{F}_f(\lambda)$ be as given in (\ref{eq_lipschitz_func_set-up}) and (\ref{eq_lipschitz_func}), where $H_0$ is a lower semi-bounded, self adjoint operator and $T_1$ and $T_2$ are two positive and bounded operators on a Hilbert space $\mathcal{H}$. In particular, for a fixed closed interval $[-p,p]$, the one-parameter family of operators $H_\lambda$ with $\lambda \in [-p,p]$, is then uniformly lower semi-bounded with spectra $\sigma(H_\lambda)$ contained in $[a + 1, + \infty)$, for some $a \in \mathbb{R}$, and $R_\lambda(a) = (H_\lambda - a)^{-1}$ exists for each $\lambda \in [-p,p]$. Given this set-up, we claim:

\begin{prop} \label{prop_lipschitz_traceclass_modif3}
Given $p >0$ and $a \in \mathbb{R}$ such that
\begin{equation} \label{eq_prop_lipschitz_traceclass_modif3_cond_spectrum}
\inf_{\lambda \in [-p,p]} \mathrm{dist}(a; \sigma(H_\lambda)) \geq 1 ~\mbox{.}
\end{equation}
Suppose that for some $m >0$, the operator $R_\lambda(a)^\frac{m}{4} T_1 \in \mathcal{S}_1$ and $T_2^2 R_\lambda(a)^\frac{m}{4} \in \mathcal{S}_1$, for each $\lambda \in [-p,p]$, so that
\begin{align}\label{eq:kappas1}
\kappa_1^{(m)}(p)&:= \sup_{\lambda \in [-p,p]} \Vert R_\lambda(a)^\frac{m}{4} T_1 \Vert_{\mathcal{S}_1} < \infty ~\mbox{,} \nonumber \\
\kappa_2^{(m)}(p)&:= \sup_{\lambda \in [-p,p]} \Vert T_2^2 R_\lambda(a)^\frac{m}{4} \Vert_{\mathcal{S}_1} < \infty ~\mbox{.}
\end{align}
Then, for each $f \in Lip_c(\mathbb{R})$ with $\mathrm{supp}(f) \subseteq [-r,r]$ and $r \geq 1$, one has that for all $\lambda_1, \lambda_2 \in [-p,p]$:
\begin{equation}
\left\vert \mathcal{F}_f(\lambda_1) -  \mathcal{F}_f(\lambda_2) \right\vert \leq (1 + m) (r+ \vert a \vert)^m \min\{\kappa_1^{(m)}(p) \Vert T_2^2 \Vert , \kappa_2^{(m)}(p) \Vert T_1 \Vert \} ~\Vert f \Vert_{\mathrm{Lip}} ~\cdot \vert \lambda_1 - \lambda_2 \vert ~\mbox{.}
\end{equation}
\end{prop}

\begin{proof}
1. Using the same approximation argument as in the proof of Proposition \ref{prop_lipschitz_traceclass}, it suffices to consider a real-valued function $f \in \mathcal{C}_c^\infty(\mathbb{R})$ with $\mathrm{supp}(f) \subseteq [-r,r]$. Set $F(s):= (s - a)^m f(s) \in Lip_c(\mathbb{R})$, which we note is still real-valued. We use the Helffer-Sj\"ostrand functional calculus to represent
\begin{align}
\mathcal{F}_f(\lambda) & = \mathrm{Tr} \{ T_2^2 R_\lambda(a)^m F(H_\lambda) \} \nonumber \\
& = \dfrac{1}{\pi} \int_{\mathbb{R}^2} \partial_{\bar z} \widetilde{F}(x,y) ~\mathrm{Tr}\{ T_2 R_\lambda(a)^\frac{m}{2} R_\lambda(z) R_\lambda(a)^\frac{m}{2} T_2 \} ~\mathrm{d}x ~ \mathrm{d}y ~\mbox{,} \label{eq_integrand_symm}
\end{align}
where we take the degree of the almost analytic extension of $F$ to be $2$. We note that, using $[R_\lambda(a), R_\lambda(z)] =0$ for $z \in \mathbb{R} \setminus \mathbb{C}$, we symmetrized the integrand in (\ref{eq_integrand_symm}) by splitting the powers of $R_\lambda(a)^m$. Similar to the proof of Proposition \ref{prop_lipschitz_traceclass}, see also Remark \ref{rem_positivemeas_1}, this will allow to recast (parts of) the derivative of $\mathcal{F}_f(\lambda)$ in terms of {\em{positive}} (instead of merely complex) measures, which will in turn improve the Lipschitz constants in the final result of Proposition \ref{prop_lipschitz_traceclass_modif3}.
Differentiating (\ref{eq_integrand_symm}) with respect to $\lambda$, one obtains
\begin{align} \label{eq_finallipschitzderiv}
\dfrac{\mathrm{d}}{\mathrm{d} \lambda} \mathcal{F}_f(\lambda) & = - \dfrac{1}{\pi} \int_{\mathbb{R}^2} \partial_{\bar z} \widetilde{F}(x,y) \left( D_\lambda^{(1)}(x,y) + D_\lambda^{(2)}(x,y) \right) ~\mathrm{d}x\mathrm{d}y ~\mbox{,}
\end{align}
with
\begin{align}
D_\lambda^{(1)}(x,y) & = \mathrm{Tr}\{ T_2 R_\lambda(a)^\frac{m}{2} R_\lambda(z) \cdot T_1 \cdot R_\lambda(z) R_\lambda(a)^\frac{m}{2} T_2 \}  ~\mbox{,} \label{eq_finallipschitz_d1} \\
D_\lambda^{(2)}(x,y) & = \sum_{k=1}^m \mathrm{Tr}\{ T_2^2 R_\lambda(a)^{m-k+1} T_1 R_\lambda(a)^k R_\lambda(z) \} ~\mbox{.} \label{eq_finallipschitz_d2}
\end{align}
We will treat the terms containing $D_\lambda^{(1)}(x,y)$ and $D_\lambda^{(2)}(x,y)$ separately.

\noindent
2. For the summand including $D_\lambda^{(1)}(x,y)$, using cyclicity of the trace, we first rewrite
\begin{align}
D_\lambda^{(1)}(x,y) = \mathrm{Tr}\{ T_2^2 R_\lambda(a)^\frac{m}{2} ~\cdot R_\lambda(z) ~\cdot T_1  R_\lambda(a)^\frac{m}{2} ~\cdot R_\lambda(z) \} ~\mbox{,}
\end{align}
which, taking advantage of the trace-class hypotheses of the proposition and of (\ref{eq_prop_lipschitz_traceclass_modif3_cond_spectrum}) shows that the trace is well-defined with
\bea
\vert D_\lambda^{(1)}(x,y) \vert &  \leq &  \Vert R_\lambda(a)^\frac{m}{2} T_1 \Vert_{\mathcal{S}_1} ~\Vert T_2^2 R_\lambda(a)^\frac{m}{2}  \Vert_{\mathcal{S}_1}  \dfrac{1}{(\mathrm{Im}(z))^2} \nonumber  \\
  & \leq &   \Vert R_\lambda(a)^\frac{m}{4} T_1 \Vert_{\mathcal{S}_1} ~\Vert T_2^2 R_\lambda(a)^\frac{m}{4}   \Vert_{\mathcal{S}_1}  \left( \dfrac{1}{ \mathrm{dist} (a; \sigma(H_\lambda) ) } \right)^{\frac{m}{2}} ~
\dfrac{1}{ ( \mathrm{Im}(z) )^2 } \nonumber \\
 & \leq &  \kappa_1^{(m)}(p) \kappa_2^{(m)}(p)    \dfrac{1}{(\mathrm{Im}(z))^2} ,
\eea
using definitions \eqref{eq:kappas1}.
In analogy to the proof of Proposition \ref{prop_lipschitz_traceclass}, we consider the bilinear functional on $\mathcal{C}_c(\mathbb{R}) \times \mathcal{C}_c(\mathbb{R})$ given by
\begin{equation}
\beta_\lambda(f,g) = \mathrm{Tr}(  T_2 R_\lambda(a)^\frac{m}{2} f(H_\lambda) \cdot T_1 \cdot g(H_\lambda) R_\lambda(a)^\frac{m}{2} T_2 ) ~\mbox{.}
\end{equation}
By the spectral theorem for compact operators, one has
\begin{align}
R_\lambda(a)^\frac{m}{2} T_1 & = \sum_{l \in \mathbb{N}} \lambda_l^{(1)} \vert \phi_l^{(1)} \rangle \langle \eta_l^{(1)} \vert ~\mbox{,} \\
T_2^2 R_\lambda(a)^\frac{m}{2} & = \sum_{l \in \mathbb{N}} \lambda_l^{(2)} \vert \phi_l^{(2)} \rangle \langle \eta_l^{(2)} \vert ~\mbox{,}
\end{align}
for appropriate orthonormal bases $\{\phi_l^{(j)}, l \in \mathbb{N} \}$ and $\{\eta_l^{(j)}, l \in \mathbb{N} \}$ and $j =1,2$.  Consequently, we can represent the bilinear form $\beta_\lambda(f,g)$ in terms of a complex Baire measure $\mu_\lambda$ on $\mathbb{R}^2$,
\begin{equation}
\beta_\lambda(f,g) = \int_{\mathbb{R}^2} f(s) g(t) \mathrm{d} \mu_\lambda(s,t) ~\mbox{,}
\end{equation}
with the measure $\mu_\lambda$ given by
\beq
\mathrm{d} \mu_\lambda(s,t)  :=  \sum_{l,j \in \mathbb{N}} \lambda_l^{(2)} \lambda_j^{(1)} \mathrm{d} \mu_{\eta_l^{(2)}; \phi_j^{(1)}}(s) \otimes \mathrm{d} \mu_{\eta_j^{(1)}; \phi_l^{(2)}}(t) .
\eeq
The total variation $|\mu_\lambda|$ is bounded by
\bea
\vert \mu_\lambda \vert & \leq & 2 \Vert R_\lambda(a)^\frac{m}{2} T_1 \Vert_{\mathcal{S}_1} ~\Vert T_2^2 R_\lambda(a)^\frac{m}{2} \Vert_{\mathcal{S}_1}  \nonumber  \\
& \leq & 2 \kappa_1^{(m)}(p) \kappa_2^{(m)}(p) ~\mbox{.} \label{eq_boundtotalvarfinallipschitz}
\eea
In particular, we thus obtain
\begin{align}
- \dfrac{1}{\pi} \int_{\mathbb{R}^2} \partial_{\bar z} \widetilde{F}(x,y) D_\lambda^{(1)}(x,y) ~\mathrm{d}x\mathrm{d}y & = \int_{\mathbb{R}^2} \left\{ \dfrac{-1}{\pi} \int_{\mathbb{R}^2} \dfrac{\partial_{\bar z} \widetilde{F}(x,y)}{(s-z)(t-z)}  ~\mathrm{d} x ~\mathrm{d}y  \right\} ~\mathrm{d}\mu_\lambda(s,t) ~\mbox{.}
\end{align}
Notice that the positivity of $T_1$ and the symmetry of the expression for $D_1(x,y)$ in (\ref{eq_finallipschitz_d1}), implies that for each function $f \geq 0$, $f \in \mathcal{C}_c(\mathbb{R})$, one has
\begin{equation}
\mu_\lambda(f,f) = \mathrm{Tr}\{ T_2 R_\lambda(a)^\frac{m}{2} f(x) \cdot T_1 \cdot f(y) R_\lambda(a)^\frac{m}{2} T_2 \} \geq 0 ~\mbox{.}
\end{equation}
This shows that $\mu_\lambda$ is positive on all cubes in $\mathbb{R}^2$, which, by the same type of argument as at the end of item 2. in the proof of Proposition \ref{prop_lipschitz_traceclass}, reveals that $\mu_\lambda$ is in fact a {\em{positive}} measure on $\mathbb{R}^2$ with total mass bounded by
\begin{align}
\mu_\lambda(\mathbb{R}^2) & = \mathrm{Tr}\{ T_2 R_\lambda(a)^\frac{m}{2} \cdot T_1 \cdot R_\lambda(a)^\frac{m}{2} T_2 \}  \nonumber  \\
& \leq \min\{ \Vert T_2^2 R_\lambda(a)^\frac{m}{2} \Vert \Vert T_1 R_\lambda(a)^\frac{m}{2} \Vert_{\mathcal{S}_1} , \Vert T_2^2 R_\lambda(a)^\frac{m}{2} \Vert_{\mathcal{S}_1}\Vert T_1 R_\lambda(a)^\frac{m}{2} \Vert \} \nonumber  \\
& \leq \min\{ \Vert T_2^2 \Vert \Vert T_1 R_\lambda(a)^\frac{m}{4} \Vert_{\mathcal{S}_1} , \Vert T_2^2 R_\lambda(a)^\frac{m}{4} \Vert_{\mathcal{S}_1}\Vert T_1 \Vert \} ~\mbox{.} \label{eq_prop_lipschitz_traceclass_modif3_totalmass_1}
\end{align}
Thus, by the Cauchy-Pompieu formula and the mean-value theorem, an argument analogous to (\ref{eq_cauchypompie_deriv_1}) - (\ref{eq_prop_lipschitz_diffcomp}) yields
\begin{align}
\left\vert \dfrac{1}{\pi} \int_{\mathbb{R}^2} \partial_{\bar z} \widetilde{F}(x,y) \right. & \left. D_\lambda^{(1)}(x,y) ~\mathrm{d}x ~\mathrm{d}y \right\vert \leq \Vert F^\prime \Vert_\infty \mu_\lambda(\mathbb{R}^2)  \nonumber  \\
& \leq (r+ \vert a \vert)^m \min\{\kappa_1^{(m)}(p) \Vert T_2^2 \Vert , \kappa_2^{(m)}(p) \Vert T_1 \Vert \} ~\Vert f \Vert_{\mathrm{Lip}} ~\mbox{,} \label{eq_finallipschitz_finalbound1}
\end{align}
where we used the bound in (\ref{eq_prop_lipschitz_traceclass_modif3_totalmass_1}) and that $r \geq 1$.

\noindent
3. To bound the term in (\ref{eq_finallipschitzderiv}) containing $D_\lambda^{(2)}(x,y)$, we analyze the contributions of each of the summands $1 \leq k \leq m$ in (\ref{eq_finallipschitz_d2}).
For $\frac{m}{2} \leq k \leq m$, we rewrite
\begin{align} \label{eq_finallipschitz_decompd2}
\mathrm{Tr}\{ T_2^2 R_\lambda(a)^{m-k+1} T_1 R_\lambda(a)^k R_\lambda(z) \} = \mathrm{Tr}\{ R_\lambda(a)^{\beta_k} T_2^2 R_\lambda(a)^{m-k+1} T_1 R_\lambda(a)^{k - \beta_k} R_\lambda(z) \} ~\mbox{,}
\end{align}
for $0 \leq \beta_k \leq k$ to be determined. Optimizing the exponents so that $\beta_k = k - \beta_k$, yields $\beta_k = \frac{k}{2} \geq \frac{m}{4}$, which shows that the trace in (\ref{eq_finallipschitz_decompd2}) is well-defined with
\begin{align} \label{eq_finallipschitz_upperbound_meask}
\vert \mathrm{Tr}\{ R_\lambda(a)^{\beta_k} T_2^2  \cdot~R_\lambda(a)^{m-k+1}  & \cdot~T_1 R_\lambda(a)^{k - \beta_k} \cdot~R_\lambda(z) \} \vert \nonumber \\
& \leq \dfrac{1}{\vert \mathrm{Im}(z) \vert} \cdot \min\{\kappa_1^{(m)}(p) \Vert T_2^2 \Vert , \kappa_2^{(m)}(p) \Vert T_1 \Vert \}  ~\mbox{.}
\end{align}
We mention that the estimate in (\ref{eq_finallipschitz_upperbound_meask}) uses that for each operator $A \in \mathcal{S}_1$, one has $A^* \in \mathcal{S}_1$ with $\Vert A^* \Vert_{\mathcal{S}_1} = \Vert A \Vert_{\mathcal{S}_1}$. 
For $1 \leq k \leq \frac{m}{2}$, we write
\begin{align} \label{eq_finallipschitz_decompd2_2}
\mathrm{Tr}\{ T_2^2 & R_\lambda(a)^{m-k+1} T_1 R_\lambda(a)^k R_\lambda(z) \}  \nonumber \\
 & = \mathrm{Tr}\{R_\lambda(z) ~\cdot T_2^2 R_\lambda(a)^{m-k - \beta_k +1/2} ~\cdot R_\lambda(a)^{\beta_k +1/2} T_1 ~\cdot R_\lambda(a)^{k} \} ~\mbox{.}
\end{align}
The exponents are optimized by letting $\beta_k = \frac{m-k}{2} \geq \frac{m}{4}$, in which case, arguments analogous to the case $k \geq \frac{m}{2}$ give rise to the same upper bound given in (\ref{eq_finallipschitz_upperbound_meask}).
Consequently, for all $1 \leq k \leq m$, the linear functional on $\mathcal{C}_c(\mathbb{R})$
\begin{equation}
f \mapsto \mathrm{Tr}\{ T_2^2 R_\lambda(a)^{m-k+1} T_1 R_\lambda(a)^k \cdot f(H_\lambda) \}
\end{equation}
is bounded and may thus be represented by a complex Baire measure $\mu_\lambda^{(k)}$ on $\mathbb{R}$ with total variation bounded by
\begin{equation}
\vert \mu_\lambda^{(k)} \vert \leq \min\{\kappa_1^{(m)}(p) \Vert T_2^2 \Vert , \kappa_2^{(m)}(p) \Vert T_1 \Vert \} ~\mbox{.}
\end{equation}
Thus, we may compute
\begin{align}
- \dfrac{1}{\pi} \int_{\mathbb{R}^2} \partial_{\bar z} & \widetilde{F}(x,y) \mathrm{Tr}\{ T_2^2 R_\lambda(a)^{m-k+1} T_1 R_\lambda(a)^k R_\lambda(z) \} ~\mathrm{d}x ~\mathrm{d}y \nonumber  \\
   & = \int_{\mathbb{R}} \left\{ \dfrac{-1}{\pi} \int_{\mathbb{R}^2} \dfrac{\partial_{\bar z} \widetilde{F}(x,y)}{(s-z)}  ~\mathrm{d} x ~\mathrm{d}y  \right\} ~\mathrm{d}\mu_\lambda^{(k)}(s)  = - \int_{\mathbb{R}} F(s) ~\mathrm{d}\mu_\lambda^{(k)}(s) ~\mbox{,}
\end{align}
which yields the bound
\begin{align}
\left\vert \dfrac{1}{\pi} \int_{\mathbb{R}^2} \partial_{\bar z} \right. & \left. \widetilde{F}(x,y) \mathrm{Tr}\{ T_2^2 R_\lambda(a)^{m-k+1} T_1 R_\lambda(a)^k R_\lambda(z) \} ~\mathrm{d}x\mathrm{d}y \right\vert \leq \Vert F \Vert_\infty ~\vert \mu_\lambda^{(k)} \vert \nonumber \\
& \leq \min\{\kappa_1^{(m)}(p) \Vert T_2^2 \Vert , \kappa_2^{(m)}(p) \Vert T_1 \Vert \} (r+ \vert a \vert)^m ~\Vert f \Vert_\infty ~\mbox{.} \label{eq_finallipschitz_finalbound2}
\end{align}

\noindent
4. In summary, combining the estimates in (\ref{eq_finallipschitz_finalbound1}) and (\ref{eq_finallipschitz_finalbound2}), we conclude
\begin{align}
\left \vert \dfrac{\mathrm{d}}{\mathrm{d} \lambda} \mathcal{F}_f(\lambda) \right \vert & \leq (1 + m) (r+ \vert a \vert)^m \min\{\kappa_1^{(m)}(p) \Vert T_2^2 \Vert , \kappa_2^{(m)}(p) \Vert T_1 \Vert \} ~\Vert f \Vert_{\mathrm{Lip}} ~\mbox{,}
\end{align}
which implies the claim.
\end{proof}

To see how Proposition \ref{prop_lipschitz_traceclass_modif3} relates to the Lipschitz property for the continuum random Schr\"odinger operators in (\ref{eq_contschrodop}), first observe that fixing a site $j \in \mathbb{Z}^d$, the analogue of the maps in (\ref{eq_setup_Lipschitz}) - (\ref{eq_setup_Lipschitz_1}) is given by
\begin{equation} \label{eq_setup_Lipschitz_conti_3}
\omega_j \mapsto \mathrm{Tr}\left(  \chi_0 f(H_{j^\perp} + \omega_j \phi(. - j))   \chi_0 \right) ~\mbox{,}
\end{equation}
where
\begin{equation} \label{eq_setup_Lipschitz_1_conti_3}
H_{j^\perp} = - \Delta + \sum_{n \in \mathbb{Z}^d, n \neq j} \omega_n \phi(. - n) ~\mbox{.}
\end{equation}
Here, for $j \in \mathbb{Z}^d$, we write $\omega = (\omega_j , \omega_j^\perp)$, where $\omega_j^\perp$ is the sequence obtained from $\omega$ by deleting the entry $\omega_j$.

For given dimension $d \in \mathbb{N}$, compactly supported values of the random potential $\omega \in [-C,C]^{\mathbb{Z}^d}$, and a site $j \in \mathbb{Z}^d$, taking
\begin{equation}
T_1 = \phi(. - j) ~\mbox{, } T_2 = \chi_0 ~\mbox{, } H_0 = H_{j^\perp} ~\mbox{, } \lambda = \omega_j \in [-C,C] ~\mbox{,}
\end{equation}
Lemma \ref{lemma:rel-resolv-bd1} shows that for every $m > 2 d$ (i.e. $\frac{m}{4} > \frac{d}{2}$, see Theorem \ref{thm_traceclassop}) and $\phi \in \mathcal{C}_c^{2(\frac{m}{4}-1)}(\mathbb{R}^d)$, all hypotheses of Proposition \ref{prop_lipschitz_traceclass_modif3} are met (letting, e.g., $a = -C \Vert \phi \Vert_\infty - 1$). In particular, we conclude that:
\begin{corollary}[Lipschitz property for continuum random Schr\"odinger operators on $\R^d$ with optimized regularity] \label{coro_lipschitzcont1_optreg}
Consider the continuum Schr\"odinger operator {\bf[{Cont}]} with $\phi \in \mathcal{C}_c^{k_v}(\mathbb{R}^d)$ such that $k_v > \max\{d - 2; 0\}$ and $\omega \in [-C,C]^{\mathbb{Z}^d}$, for $0 < C < +\infty$. Then, for each fixed site $j \in \mathbb{Z}^d$ and $f \in Lip_c(\mathbb{R})$ with $\mathrm{supp}(f) \subseteq [-r,r]$ and $r \geq 1$, the map in (\ref{eq_setup_Lipschitz_conti_3}) is Lipschitz with Lipschitz constant bounded above by
\begin{equation}
C_7(k_v, \Vert \phi \Vert_{\mathcal{C}^{k_v}}, C) r^{2(k_v+1)} \cdot \|f \|_{\rm Lip} ~\mbox{.}
\end{equation}
\end{corollary}


\section{Proof of the main theorems}
\setcounter{equation}{0}

We present the outline of the proofs of Theorems \ref{thm:main-lattice1} and \ref{thm:main-cont1}. Since many of the details are the same as in the proofs in \cite{hislop_marx_1}, we refer the reader to that paper for some details.
We recall the two main components for the quantitative bound on the DOSm.

\begin{itemize}

\item {\bf{Finite range reduction:}} We recall that the operator $H_{\omega;L}^{(0)}$ depends on only finitely-many random variables.
For the lattice model with single-site probability measures with finite first moments, we proved in Theorem \ref{thm_finiterangereduction}:
\begin{equation}\label{eq:finite-range-lattice1}
n_\nu^{(\infty)}(f) = \frac{1}{N} \mathbb{E}_{\nu^{(L)}} [ \mathrm{Tr}(P_0 f(H_{\omega; L}^{(0)}) P_0)] + R_f[L; \nu] ~\mbox{,}
\end{equation}
where
\begin{equation}\label{eq_thm_finite3}
\vert R_f[L; \nu] \vert \leq c_1 \dfrac{\mu_1[\nu]}{NL}  \Vert f \Vert_{3+d} ~\mbox{.}
\end{equation}
for all $f \in \mathcal{C}_c^\beta(\mathbb{R})$ with $\beta \geq 3 + d$.
For the continuum model with compactly supported single site probability measures, we proved results analogous to \eqref{eq:finite-range-lattice1}--\eqref{eq_thm_finite3} in Theorem \ref{thm_finiterangereduction-cont1}:
\begin{equation}\label{eq:finite-range-cont2}
n_\nu^{(\infty)}(f) = \mathbb{E}_{\nu^{(L)}} [ \mathrm{Tr}(\chi_0 f(H_{\omega; L}^{(0)}) \chi_0)] + R_f[L; \nu] ~\mbox{,}
\end{equation}
where
\begin{equation}\label{eq_thm_finite_cont3}
\vert R_f[L; \nu] \vert \leq  \dfrac{c_4}{L}  \Vert f \Vert_{3+d;m} ~\mbox{,}
\end{equation}
where $\mu_{max} [\nu] := \max_{k=1, \ldots, m-1}\mu_k[\nu]$ and $m > \frac{d}{2}$.

\item {\bf{Lipschitz property:}} For the lattice model with single-site probability measures with finite first moments,
we concluded from 
Proposition \ref{prop_lipschitz_traceclass} that, for every $f \in Lip_c(\R)$, the map
\beq
\omega_n \mapsto {\rm Tr} (P_0 f(H_{\omega_n \omega_n^\perp;L}^{(0)}) P_0), ~~ n \in \{ j \in \Z^d ~|~ \| j \|_\infty \leq K L \} ,
\eeq
is bounded Lipschitz with Lipschitz constant $\leq N L_f$ (see Corollary \ref{coro_lipschitzdiscrete}), where
$N$ is the rank of the projection $P_0$. For the continuum mode  with compactly supported single site probability measuresl, Corollary \ref{coro_lipschitzcont1_optreg} obtains an analogous statement for the map  
\beq\label{eq:cont-lip-summ1}
\omega_n \mapsto {\rm Tr} (P_0 f(H_{\omega_n \omega_n^\perp;L}^{(0)}) P_0), ~~ n \in \{ j \in \Z^d ~|~ \| j \|_\infty \leq L \},
\eeq
for $P_0 = \chi_0$ and for every $f \in Lip_c(\R)$. In this case the Lipschitz constant is bounded above by $C_7 r^{2 (k_v +1)} \Vert f \Vert_{Lip}$, where $k_v > d+2$ and $r \geq 1$ is such that $\mathrm{supp}(f) \subseteq [ - r, r]$.
\end{itemize}

As can be seen, the structure of the finite range reduction and the Lipschitz property are the same in both the discrete and the continuum case. We will thus only present a proof of Theorem \ref{thm:main-lattice1}; Theorem \ref{thm:main-cont1} is proven in complete analogy.

\subsection{Proof of Theorem \ref{thm:main-lattice1} for lattice models}

We start by proving part (ii.) of Theorem \ref{thm:main-lattice1} about the modulus of continuity of the DOSm. The density of $\mathcal{C}_c^{M_1}(\mathbb{R})$ in $Lip_c(\mathbb{R})$ then implies the qualitative continuity statement for the DOSm in part (i.).

First, observe that, for $k,l \in \mathbb{N} \cup \{0 \}$, given $f \in \mathcal{C}_c^k(\mathbb{R})$ with $\mathrm{supp}(f) \subseteq [-r,r]$, the definition of the norm $\Vert f \Vert_{k;l}$ implies
\begin{equation} \label{eq_boundfklfinal}
\Vert f \Vert_{k;l} \leq r^{k+l} \Vert f \Vert_{\mathcal{C}^k} ~\mbox{,}
\end{equation}
with $\Vert f \Vert_{\mathcal{C}^k}$ as given in (\ref{eq_cbetanorm}).

Turning to the proof, given $C>0$, let $\nu_1 \neq \nu_2 \in \mathcal{P}_{1;C}(\mathbb{R})$ as defined in (\ref{eq_probspace_boundedfinitemoment}), i.e. $\mu_1[\nu_j] \leq C$, for $j=1,2$. Write $\epsilon:= d_w(\nu_1, \nu_2) > 0$.

We take $M_1 = 3+d$ in Theorem \ref{thm:main-lattice1}. For $f \in \mathcal{C}^{3 + d}_c (\R)$, ${\mbox{supp}} ~f \subseteq [- r, r]$, and $L \in \N$ arbitrary, the definition of the DOSm, the finite range reduction \eqref{eq:finite-range-lattice1}--\eqref{eq_thm_finite3}, and the bound on $\| f \|_{3+d}=\Vert f \Vert_{3+d;0}$ in (\ref{eq_boundfklfinal}) give
\bea\label{eq:dosm-est1}
  | n_{\nu_1}^{(\infty)} (f) - n_{\nu_2}^{(\infty)} (f) |  & \leq & \frac{1}{N} \left| \E_{\nu_1^{(L)}}  {\rm Tr} ~ \left( P_0 f(H_{\omega;L}^{(0)}) P_0 \right) -   \E_{\nu_2^{(L)}}  {\rm Tr} ~ \left( P_0 f(H_{\omega;L}^{(0)}) P_0 \right) \right| \nonumber \\
 & & + \frac{2 C c_1}{NL} \| f \|_{3+d}   \nonumber \\
 & \leq  & 2 (2L+1)^d N L_f \cdot \epsilon +  \frac{2 C c_1}{NL} r^{3+d} \| f \|_{\mathcal{C}^{3+d}} \label{eq_final_1}  \\
 & \leq & \frac{C_1}{6} r^{3 +d} \Vert f \Vert_{\mathcal{C}^{3+d}} \left(3 L^d \cdot \epsilon + \frac{1}{L} \right)  \label{eq_final_2} ~\mbox{.}
\eea
for $C_1>0$ depending on $C$, $d$, and $N$. We note that we used the Lipschitz property in the first term in (\ref{eq_final_1}).

We set
\begin{equation} \label{eq_defnrho1}
\rho_1:= \left(\frac{2}{3}\right)^{1+d} ~\mbox{,}
\end{equation}
and suppose that $\epsilon= d_w(\nu_1, \nu_2) < \rho_1$. Since $L \in \mathbb{N}$ was arbitrary until now, we can choose $L \in \mathbb{N}$ such that
\begin{equation} \label{eq_finalL}
\frac{1}{3} \epsilon^{-\xi} \leq L < \epsilon^{-\xi} ~\mbox{,}
\end{equation}
for $\frac{1}{1+d} \leq \xi$ to be determined later. Notice that since $\epsilon= d_w(\nu_1, \nu_2) < \rho_1$, the definition of $\rho_1$ and $\frac{1}{1+d} \leq \xi$ implies that $L \in \mathbb{N}$ as in (\ref{eq_finalL}) exists.

From (\ref{eq_final_2}), we thus conclude
\begin{equation}
 | n_{\nu_1}^{(\infty)} (f) - n_{\nu_2}^{(\infty)} (f) | \leq \frac{C_1}{2} r^{3 +d} \Vert f \Vert_{\mathcal{C}^{3+d}} ( \epsilon^{-\xi d + 1} + \epsilon^{\xi} ) ~\mbox{,}
\end{equation}
which, optimizing in $\xi$, yields $\xi=\frac{1}{1+d}$. In summary, we obtain the claimed modulus of continuity for the DOSm in Theorem \ref{thm:main-lattice1} (ii), i.e.
\begin{equation}
  | n_{\nu_1}^{(\infty)} (f) - n_{\nu_2}^{(\infty)} (f) | \leq C_1 r^{3 +d} \Vert f \Vert_{\mathcal{C}^{3+d}} ~d_w(\nu_1, \nu_2)^{\frac{1}{1 + d}} ~\mbox{.}
\end{equation}

\begin{remark}
In our previous work for the lattice model with compactly supported single-site probability measures, we obtained the smaller H\"older exponent $(1 + 2d)^{-1}$ but the method allowed us to treat treat less regular functions $f \in Lip_c(\R)$.
\end{remark}

The proof of Theorem \ref{thm:main-lattice1} (iii.) for the IDS is similar to the one presented for Theorem 3.2 in \cite{hislop_marx_1} requiring the approximation of the step function (see also remark \ref{rem_mainthm_disc} (ii.)) and optimization. As in \cite{hislop_marx_1}, this uses that the IDS is $\log$-H\"older continuous in $E$, see \cite{craig-simon1}, i.e. for constants $C>0$, $A \in \mathbb{R}$ and each fixed measure $\nu \in \mathcal{P}_{1;C}([A,+ \infty))$, one has that for all $E_0 \in \mathbb{R}$, there exist $K_{d;C;A; E_0}$ such that for all $E \leq E_0$ and all $0 < \epsilon < \frac{1}{2}$, one has
\begin{equation} \label{eq_IDSdisc}
\vert N_\nu(E) - N_\nu(E + \epsilon) \vert \leq \dfrac{K_{d;C; A; E_0} }{\log (\frac{1}{\epsilon})}~\mbox{.}
\end{equation}
We note that the $E_0$ dependence of the constant takes into account that the potential is unbounded since $\nu$ has unbounded support.

For the continuum model, the analogue of (\ref{eq_IDSdisc}), for dimensions $d=1,2,3$, was established by Bourgain and Klein \cite{bourgain-klein1}, in which case (\ref{eq_IDSdisc}) is replaced by
\begin{equation}\label{eq_IDScont}
\vert N_\nu(E) - N_\nu(E + \epsilon) \vert \leq \dfrac{K_{d;C;E_0} }{[ \log (\frac{1}{\epsilon}) ]^{\kappa_d}} ~\mbox{, where } \kappa_1=1 ~,~ \kappa_2 = \frac{1}{4} ~,~ \kappa_3 = \frac{1}{8} ~\mbox{.}
\end{equation}

\section{Appendix 1: Almost analytic extensions and the Helffer-Sj\"ostrand functional calculus}\label{sec:appendix-helffer-sj1}

In what follows, we will rely on some aspects of the theory of almost analytic extensions of functions on the real line and the Helffer-Sj\"ostrand functional calculus. For reference purposes, we briefly summarize some facts here which will be of use to us; for a more detailed and pedagogical account, we refer e.g. to \cite{davies1}.

Let $\tau \in \mathcal{C}^\infty$ be a fixed bump function on $\mathbb{R}$ with support in $[-2,2]$, satisfying $\tau \equiv 1$ on $[-1,1]$. Set
\begin{equation} \label{eq_rescaledBump}
\sigma(x,y) : = \tau\left( \dfrac{y}{\langle x \rangle}   \right) ~\mbox{.}
\end{equation}

Given a (complex-valued) function $f \in \mathcal{C}_c^\infty(\mathbb{R})$ and an integer $P \in \mathbb{N}$, one can define an {\em{almost analytic extension}} $\widetilde{f}$ of $f$ {\em{of degree}} $P$ by
\begin{equation} \label{eq_almostanalyticextension}
\widetilde{f}(x,y) = \left\{ \sum_{n=0}^P \frac{1}{n !} f^{(n)}(x) (i y)^n  \right\} \sigma(x,y) ~\mbox{.}
\end{equation}
Then, $\widetilde{f} \equiv f$ on $\mathbb{R}$, $\widetilde{f}$ is compactly supported on $\mathbb{C}$, and a straight-forward computation shows that
\begin{align} \label{eq_almostanalytic_1}
\partial_{\overline{z}} \widetilde{f}:= \frac{1}{2} ( \partial_x + i \partial_y ) \widetilde{f} = \frac{1}{2} \left\{ \sum_{n=0}^P \frac{1}{n !} f^{(n)}(x) (i y)^n  \right\} (\sigma_x + i \sigma_y) + \frac{1}{2} \frac{1}{P!} f^{(P+1)}(x) (iy)^P \sigma ~\mbox{.}
\end{align}
In particular, using the properties of $\sigma$, (\ref{eq_almostanalytic_1}) implies that $\widetilde{f}$ is almost analytic in a neighborhood of $\mathbb{R}$, in the sense that
\begin{equation} \label{eq_almostanalytic}
\vert \partial_{\overline{z}} \widetilde{f}(x,y) \vert \leq {\Vert f^{(P+1)} \Vert_\infty} ~{\vert y \vert^P} ~\mbox{, } ~ \vert y \vert \leq \langle x \rangle ~\mbox{.}
\end{equation}

One important application of almost analytic extensions of functions on $\mathbb{R}$ is an explicit representation of functions of operators via the {\em{Helffer-Sj\"ostrand functional calculus}}: For every self-adjoint operator $H$ and $f \in \mathcal{C}_c(\mathbb{R})$, one has the representation
\begin{equation} \label{eq_helffersjostrand}
f(H) = \dfrac{1}{\pi} \int\int_\mathbb{C} \partial_{\overline{z}} \widetilde{f}(x,y) \cdot (H -  z)^{-1} ~\mathrm{d} x \mathrm{d} y ~\mbox{.}
\end{equation}
It can be shown (see e.g. \cite{davies1}, Chapter 2.2) that this representation is well-defined in the sense that $f(H)$ is independent of the bump function $\tau$ and the degree $P \geq 1$ used to define the almost analytic extension $\widetilde{f}$ in (\ref{eq_almostanalyticextension}).

%
%
%

\end{document}